\documentclass[submission,copyright,creativecommons]{eptcs}
\providecommand{\event}{MSFP 2016} 
\usepackage{breakurl}             

\usepackage{amssymb}
\usepackage{latexsym}
\usepackage{amsmath}

\usepackage[dvips]{graphicx}

\usepackage{stmaryrd}

\def\max{\mathop{\rm max} \nolimits}

\def\mathdef{\mathop{\rm def} \nolimits}

\newcommand{\MBB}{\mathbb}

\newcommand{\PAR}{\bindnasrepma}
\newcommand{\TENS}{\otimes}

\def\linearimp{\mathop{- \hspace{-.025in} \circ}}

\newcommand{\LIMP}{\linearimp}

\newtheorem{lemma}{Lemma}[section]
\newtheorem{theorem}{Theorem}[section]
\newtheorem{definition}{Definition}[section]
\newtheorem{corollary}{Corollary}[section]
\newtheorem{proposition}{Proposition}[section]
\newtheorem{example}{Example}[section]
\newtheorem{remark}{Remark}

\newenvironment{proof}{\begin{flushleft}{\bf Proof:} \ \ }{\end{flushleft}}

\newenvironment{ack}{\begin{flushleft}{\bf Acknowledgments.} \ \ }{\end{flushleft}}

\title{Strong Typed B\"{o}hm Theorem and Functional Completeness on the Linear Lambda Calculus}
\author{Satoshi Matsuoka\\
\institute{National Institute of Advanced Industrial Science and Technology (AIST),}
\institute{1-1-1 Umezono, Tsukuba, Ibaraki, 305-8565 Japan}
\email{matsuoka@ni.aist.go.jp}
  }

\def\titlerunning{Strong Typed B\"{o}hm Theorem}
\def\authorrunning{S. Matsuoka}
\begin{document}
\maketitle
\begin{abstract}
In this paper, we prove a version of the typed B\"{o}hm theorem on the linear lambda calculus,
which says, for any given types $A$ and $B$,
when two different closed terms $s_1$ and $s_2$ of $A$ and 
any closed terms $u_1$ and $u_2$ of $B$ are given, 
there is a term $t$ such that
$t \, s_1$ is convertible to $u_1$ and $t \, s_2$ is convertible to $u_2$.
Several years ago, a weaker version of this theorem was proved, but the stronger version was open. 
As a corollary of this theorem, we prove that 
if $A$ has two different closed terms $s_1$ and $s_2$, then 
$A$ is functionally complete with regard to $s_1$ and $s_2$.
So far, it was only known that a few types are functionally complete.
\end{abstract}
\section{Introduction}
This paper is an addendum to the paper \cite{Mat07}, which was published several years ago.
The previous paper establishes the following result in the linear $\lambda$-calculus:
\begin{quote}
For any type $A$ and two different closed terms $s_1$ and $s_2$ of type $A$, 
there is a term $t$ such that 
\[
t s_1 =_{\beta \eta {\rm c}} \underline{0} \quad \mbox{and} \quad t s_2 =_{\beta \eta {\rm c}} \underline{1} \, ,
\]
where $\underline{0} \equiv_{\mathdef} \lambda x. \lambda f. \lambda g. f(g(x))$ and $\underline{1} \equiv_{\mathdef} \lambda x. \lambda f. \lambda g. g(f(x))$. 
\end{quote}
In \cite{Mat07}, the proof net notation for the intuitionistic multiplicative linear logic (for short, IMLL) was used, but as shown later,
the linear $\lambda$-calculus can be regarded as a subsystem of 
IMLL proof nets.
In addition the equality $=_{\beta \eta {\rm c}}$ will be defined precisely later. 
In this paper, we prove a stronger version of the previous statement, which is stated as follows:
\begin{quote}
  For any given types $A$ and $B$,
  when two different closed terms $s_1$ and $s_2$ of $A$ and 
  any closed terms $u_1$ and $u_2$ of $B$ are given, 
  there is a term $t$ such that 
\[
t \, s_1 =_{\beta \eta {\rm c}} u_1 \quad \mbox{and} \quad t \, s_2 =_{\beta \eta {\rm c}} u_2 \, .
\]
\end{quote}
The stronger version was an open question in \cite{Mat07}.
Note that the strong version is trivially derived from the weak one in the simply typed $\lambda$-calculus, because
the calculus allows discard and copy of variables freely.
But the linear $\lambda$-calculus officially does not allow these two operations. 
So some technical devices are required. 
The basic idea of our solution is 
to extend the typability by a linear implicational formula $A \LIMP B$ to a more liberalized form.
We call the extended typability {\it poly-typability},
which is a mathematical formulation of the typing discipline used in \cite{Mai04}. 
Thanks to the extension, we can prove 
Projection Lemma (Lemma~\ref{lemmaProjectionLemma}) and Constant Function Lemma (Lemma~\ref{lemmaConstantFunctionLemma}),
which are the keys to establish our typed B\"{o}hm theorem. 

One application is the functional completeness problem of the linear $\lambda$-calculus.
It raises the question about the possibility of Boolean representability in the linear $\lambda$-calculus.
We prove that any type with at least two different closed terms is functionally complete.
This means that any two-valued functions can be represented over these two terms. 
So far, it was only known that a few types have this property.
Our functional completeness theorem liberalizes us from sticking to specific types.
This situation is analogous to that of the degree of freedom about a base choice in linear algebra: linear independence is enough. 
Similarly we may choose any different two terms of any type in order to establish the functional completeness. 

The strong typed B\"{o}hm theorem gives a general construction of linear $\lambda$-terms that satisfy a given specification for inputs and outputs. 
It is expected that useful theorems about linear $\lambda$-terms will be proved by using the theorem further.

\paragraph{Comparison with the case of the simply typed lambda calculus}
The first proof of the typed B\"{o}hm theorem for the simply typed lambda calculus
was given in \cite{Sta83}.
The proof is based on the {\it reducibility theorem} in \cite{Sta80} (see also Theorem 3.4.8 in \cite{BDS13}).
Our proof proceeds in a similar manner to Statman's proof.
But the proof of the reducibility theorem is rather complicated,
since it uses different operations.
On the other hand, 
the proof of our analogue, which is
Proposition~\ref{propLDT-basic}, is much simpler,
because our proof is based on one simple principle, i.e., linear distributive law (see, e.g.,  \cite{BS04})
\footnote{For example, this principle includes
  $(((A \LIMP B) \LIMP C) \LIMP D) \LIMP (A \LIMP (B \LIMP C) \LIMP D)$,
  $((A \LIMP B) \TENS C) \LIMP (A \LIMP (B \TENS C))$,
  and
  $((A \LIMP B \TENS C) \LIMP D) \LIMP (B \LIMP (A \LIMP C) \LIMP D)$.
  This observation was the starting point of Proposition~\ref{propLDT-basic}.}:
\[
((A \PAR B) \TENS C) \LIMP (A \PAR (B \TENS C))
\]
On the other hand, while the final separation argument of Statman's proof only uses type instantiation,
our proof of Theorem~\ref{thmStrongTypedBohmTheorem} needs the notion of poly-types. 

\section{Typing Rules, Reduction Rules, and an Equational Theory}
In this section we give our type assignment system for the linear $\lambda$-calculus and discuss 
some reduction rules and equivalence relations on the typed terms of the system.
Our system is based on the natural deduction calculus given in \cite{Tro92},
which is equivalent to the system based on the sequent calculus or proof nets in \cite{Gir87}
(e.g., see \cite{Tro92}). 
Our notation is the same as that in \cite{Mai04}:
the reader can confirm our results using an implementation of Standard ML \cite{MTHM97}.
\paragraph{Types} 
\[ \mbox{\tt A} ::= \mbox{\tt 'a} \, \, \, \, | \, \, \, \, \mbox{\tt A1*A2} \, \, \, \, | \, \, \, \, \mbox{\tt A1->A2} \]
The symbol {\tt 'a} stands for a type variable.
On the other hand $\mbox{\tt A1*A2}$ stands for the tensor product $\mbox{\tt A1} \TENS \mbox{\tt A2}$ and $\mbox{\tt A1->A2}$ for the linear implication $\mbox{\tt A1} \LIMP \mbox{\tt A2}$
in the usual notation.
\paragraph{Terms}
We use {\tt x,y,z} for term variables 
and {\tt r,s,t,u,v,w} for general terms.
\paragraph{Linear Typing Contexts}
A linear typing context is a finite list of pairs $\mbox{\tt x:A}$ such that 
each variable occurs in the list once.
Usually we use Greek letters $\Gamma, \Delta, \ldots$ to denote linear typing contexts.
\paragraph{Type Assignment System}
\[
\frac{}{\mbox{\tt x:A} \vdash \mbox{\tt x:A}}
\quad \quad
\frac{\Gamma, \mbox{\tt x:A}, \mbox{\tt y:B}, \Delta \vdash \mbox{\tt t:C}}
{\Gamma, \mbox{\tt y:B}, \mbox{\tt x:A}, \Delta \vdash \mbox{\tt t:C}}
\]
\[
\frac{\mbox{\tt x:A}, \Gamma \vdash \mbox{\tt t:B}}
{\Gamma \vdash \mbox{\tt fn x=>t:A->B}}
\quad \quad 
\frac
{ \Gamma \vdash \mbox{\tt t} : \mbox{\tt A->B} \quad \Delta \vdash \mbox{\tt s:A} }
{ \Gamma, \Delta \vdash \mbox{\tt t} \, \mbox{\tt s} \mbox{\tt:B}}
\]
\[
\frac
{ \Gamma \vdash \mbox{\tt s:A} \quad \Delta \vdash \mbox{\tt t:B} }
{ \Gamma, \Delta \vdash \mbox{\tt (s,t):A*B}}
\quad \quad 
\frac
{ \Gamma \vdash \mbox{\tt s:A*B}
\quad \mbox{\tt x:A}, \mbox{\tt y:B}, \Delta \vdash \mbox{\tt t:C} }
{ \Gamma, \Delta \vdash \mbox{\tt let val (x,y)=s in t end:C} }
\]
In addition we assume that
for each term variable, 
if an occurrence of the variable appears in a sequent in a term derivation, then 
the number of the occurrences in the sequent is exactly two. 
For a term {\tt t} the set of bound variables ${\rm BV}(\mbox{\tt t})$ is defined recursively as follows:
\begin{itemize}
  \item ${\rm BV}(\mbox{\tt x}) = \emptyset$, 
  \item ${\rm BV}(\mbox{\tt s}  \, \, \mbox{\tt t})
  =
  {\rm BV}(\mbox{\tt (s,t)})
  = 
  {\rm BV}(\mbox{\tt t}) 
  \cup
  {\rm BV}(\mbox{\tt s})$, 
  \item ${\rm BV}(\mbox{\tt fn x=>t}) = \{ \mbox{\tt x} \} \cup {\rm BV}(\mbox{\tt t})$, 
  \item ${\rm BV}(\mbox{\tt let val (x,y)=s in t end})
  =
  \{ \mbox{\tt x}, \mbox{\tt y} \} \cup {\rm BV}(\mbox{\tt s}) \cup {\rm BV}(\mbox{\tt t})$. 
\end{itemize}
The set of free variables of {\tt t}, denoted by ${\rm FV}(\mbox{\tt t})$  is the complement of the set of variables in {\tt t} with respect to ${\rm BV}(\mbox{\tt t})$. 
The function declaration \\
\ \ \ \ \ \ \ \ {\tt fun f x1 x2 $\cdots$ xn = t} \\
is interpreted as the following term: \\ 
\ \ \ \ \ \ \ \ {\tt f = fn x1=>fn x2=> $\cdots$ =>fn xn=>t} \\
Below we consider only closed terms (i.e. combinators) \\
\ \ \ \ \ \ \ \ {\tt $\vdash$ t:A}. 
\paragraph{Term Reduction Rules}
Two of our reduction rules are\\
\ \ \ \ \ \ \ \ ($\beta_1$): {\tt (fn x=>t)s} \ \ $\Rightarrow_{\beta_1}$ \ \  {\tt t[s/x]} \\
\ \ \ \ \ \ \ \ ($\beta_2$): {\tt let val (x,y)=(u,v) in w end} \ \ $\Rightarrow_{\beta_2}$ \ \  {\tt w[u/x,v/y]} \\
Then note that if a function {\tt f} is defined by \\ 
\ \ \ \ \ \ \ \ {\tt fun f x1 x2 $\cdots$ xn = t}\\ 
and \\
\ \ \ \ \ \ \ \ {\tt x1:A1,...,xn:An|-t:B}, $\, \, \, \,$ {\tt |-t1:A1}, $\, \,$ $\ldots$, $\, \, \, \, $ {\tt |-tn:An} \\
then, we have \\
\ \ \ \ \ \ \ \ {\tt f t1 $\cdots$ tn} \ \ $\Rightarrow_{\beta_1}^{\ast}$ \ \ {\tt t[t1/x1,$\ldots$,tn/xn]} .\\
We denote the reflexive transitive closure of a relation $R$ by $R^{\ast}$. 
In the following $\to_{\beta}$ denotes the congruent (one-step reduction) relation generated by the two reduction rules above and the following contexts:
\begin{eqnarray*}
  C[] & = & [ ] \, \, \bigm| \, \, C[] \, \mbox{\tt t} \, \, \bigm| \, \, \mbox{\tt t} \, C[] \, \,
  \bigm| \, \, \mbox{\tt (t,} C[] \mbox{\tt )} \, \, \bigm| \, \, \mbox{\tt (} C[] \mbox{\tt , t)} 
  \, \, \bigm| \, \, \mbox{\tt fn} \, \, \mbox{\tt x} \mbox{\tt =>} \, C[] \\
& & \, \,
\, \, \bigm| \, \, \mbox{\tt let} \, \,  \mbox{\tt val} \, \, \mbox{\tt (x, y) =} \, \, C[] \, \, \mbox{\tt in} \, \,  \mbox{\tt t} \, \, \mbox{\tt end}
\, \, \bigm| \, \, \mbox{\tt let} \, \,  \mbox{\tt val} \, \, \mbox{\tt (x, y) = t} \, \,  \mbox{\tt in} \, \, C[] \, \, \mbox{\tt end} 
\end{eqnarray*}
We define the set of variables captured by a context $C[]$, denoted by ${\rm CV}(C[])$ recursively:
\begin{itemize}
\item ${\rm CV}([]) = \emptyset$,
\item ${\rm CV}(C[] \, \, \mbox{\tt t})
  = {\rm CV}(\mbox{\tt t} \, \, C[])
  = {\rm CV}(\mbox{\tt (t,} C[] \mbox{\tt )})
  = {\rm CV}(\mbox{\tt (} C[] \mbox{\tt ,} \, \mbox{\tt t)})
  = {\rm CV}(C[])$,
\item ${\rm CV}(\mbox{\tt fn} \, \, \mbox{\tt x} \mbox{\tt =>} \, C[]) = \{ \mbox{\tt x} \} \cup {\rm CV}(C[])$,
\item ${\rm CV}(\mbox{\tt let val (x,y) =} \, \, C[] \, \, \mbox{\tt in} \, \,  \mbox{\tt t} \, \, \mbox{\tt end})
  = {\rm CV}(C[])$,
\item ${\rm CV}(\mbox{\tt let val (x,y) = t in} \, \, C[] \, \, \mbox{\tt end} )
  = \{ \mbox{\tt x}, \mbox{\tt y} \} \cup {\rm CV}(C[])$.
\end{itemize}
The set of free variables of a context $C[]$, denoted by ${\rm FV}(C[])$ is defined similarly to that of a term {\tt t}.

In order to establish a full and faithful embedding from linear $\lambda$-terms into IMLL proof nets, 
we introduce further reduction rules.
Basically we follow \cite{MRA93}, but note that
a simpler presentation is given than that of \cite{MRA93}, following a suggestion of an anonymous referee. 
The following are $\eta$-rules:\\
\ \ \ \ \ \ \ \ ($\eta_1$): {\tt fn x=>(t x)} \ \ $\Rightarrow_{\eta_1}$ \ \  {\tt t} \\
\ \ \ \ \ \ \ \ ($\eta_2$): {\tt let val (x,y) = t in (x,y)} \ \ $\Rightarrow_{\eta_2}$ \ \ {\tt t} \\
In the following $\to_{\beta \eta}$ denotes the congruent (one-step reduction) relation generated by the four reduction rules above and any context $C[]$.
But these reduction rules are not enough: 
different normal terms may correspond to the same normal IMLL proof net. 
In order to make further identification we introduce the following commutative conversion rule. 
Then we define the commutative conversion relation $\leftrightarrow_{\rm c}$:
\[
\begin{array}{l}
  C[  \mbox{\tt let} \, \,  \mbox{\tt val} \, \,  \mbox{\tt (x,y)=t} \, \,  \mbox{\tt in} \, \,  \mbox{\tt u} \, \,  \mbox{\tt end}] \, \, \leftrightarrow_{\rm c} \, \, \mbox{\tt let} \, \,  \mbox{\tt val} \, \,  \mbox{\tt (x,y)=t} \, \,  \mbox{\tt in} \, \,  C[ \mbox{\tt u} ] \, \, \mbox{\tt end} \\
  \mbox{where} \, \, {\rm FV}(C[]) \cap \{ \mbox{\tt x} , \mbox{\tt y}  \} = \emptyset
  \, \, \mbox{and} \, \, {\rm CV}(C[]) \cap {\rm FV}({\tt t}) = \emptyset
  \end{array}
\]
Let $=_{\rm c}$ be the congruent equivalence relation generated by $\leftrightarrow_{\rm c}$ and any context $C[]$. 
Then we define $\to_{\beta \eta {\rm c}}$ as the least relation satisfying the following rule:
\[
\frac{\mbox{\tt t} =_{\rm c} \mbox{\tt t'} \quad \mbox{\tt t'} \to_{\beta \eta} \mbox{\tt u'} \quad \mbox{\tt u'} =_{\rm c} \mbox{\tt u}}{\mbox{\tt t} \to_{\beta \eta {\rm c}} \mbox{\tt u}}
\]
Then the following holds.
\begin{proposition}[Church Rosser\cite{MRA93}]
if $\mbox{\tt t} \to_{\beta \eta {\rm c}} \mbox{\tt t'}$ and $\mbox{\tt t} \to_{\beta \eta {\rm c}} \mbox{\tt u'}$ then 
for some $\mbox{\tt w} =_{\rm c} \mbox{\tt w'}$, $\mbox{\tt t'} \to_{\beta \eta {\rm c}} \mbox{\tt w}$ and $\mbox{\tt u'} \to_{\beta \eta {\rm c}} \mbox{\tt w'}$. 
\end{proposition}
Furthermore we can easily prove that $\to_{\beta \eta {\rm c}}$ is strong normalizable as shown in \cite{MRA93}. 
We can conclude that we have the uniqueness property for normal forms under $\to_{\beta \eta {\rm c}}$ up to $=_{\rm c}$. 
\paragraph{Equality Rules}
Next we define our fundamental equality $=_{\beta \eta {\rm c}}$, which is given in \cite{MRA93} implicitly. 
The equality $=_{\beta \eta {\rm c}}$ is the smallest relation satisfying the following rules of the three groups:\\
(Relation Group)
\[
{\rm (Refl)}
\frac{\Gamma \vdash \mbox{\tt t:A}}
{\Gamma \vdash \mbox{\tt t}= \mbox{\tt t:A}}
\quad
{\rm (Sym)}
\frac{\Gamma \vdash \mbox{\tt t}= \mbox{\tt s:A}}
{\Gamma \vdash \mbox{\tt s}=\mbox{\tt t:A}}
\quad
{\rm (Trans)}
\frac{\Gamma \vdash \mbox{\tt t}= \mbox{\tt s:A} \quad \Gamma \vdash \mbox{\tt s}=\mbox{\tt u:A}}
{\Gamma \vdash \mbox{\tt t}=\mbox{\tt u:A}}
\]
(Reduction Group)
\[
({\rm Eq} {\rm c})
\frac{\Gamma \vdash \mbox{\tt t:A} \quad \mbox{\tt t} \, \leftrightarrow_{{\rm c}} \, \mbox{\tt t'}}
{\Gamma \vdash \mbox{\tt t} = \mbox{\tt t':A}}
\quad \quad
({\rm Eq} \beta \eta)
\frac{\Gamma \vdash \mbox{\tt t:A} \quad \mbox{\tt t} \, \rightarrow_{\beta \eta {\rm c}} \, \mbox{\tt t'}}
{\Gamma \vdash \mbox{\tt t} = \mbox{\tt t':A}}
\]
(Congruence Group)
\[
({\rm Eq} \lambda)
\frac{\mbox{\tt x:A}, \Gamma \vdash \mbox{\tt t} = \mbox{\tt t':B}}
{\Gamma \vdash \mbox{\tt fn x=>t} = \mbox{\tt fn x=>t':A->B}}
\]
\[
({\rm Eq} \, {\rm ap})
\frac
{ \Gamma \vdash \mbox{\tt t} = \mbox{\tt t':A->B}
\quad \Delta \vdash \mbox{\tt s} = \mbox{\tt s':A}}
{\Gamma, \Delta \vdash \mbox{\tt t} \, \mbox{\tt s} = \mbox{\tt t'} \, \mbox{\tt s'} \mbox{\tt :B}}
\]
\[
({\rm Eq} \, {\rm tup})
\frac
{ \Gamma \vdash \mbox{\tt s} = \mbox{\tt s':A} \quad \Delta \vdash \mbox{\tt t} = \mbox{\tt t':B} }
{ \Gamma, \Delta \vdash \mbox{\tt (s,t)} = \mbox{\tt (s',t'):A*B}}
\]
\[
({\rm Eq} \, {\rm let})
\frac
{ \Gamma \vdash \mbox{\tt s} = \mbox{\tt s':A*B}
\quad \mbox{\tt x:A}, \mbox{\tt y:B}, \Delta \vdash \mbox{\tt t} = \mbox{\tt t':C} }
{ \Gamma, \Delta \vdash \mbox{\tt let val (x,y)=s in t end} = \mbox{\tt let val (x,y)=s' in t' end:C}}
\]
\paragraph{The relationship between linear $\lambda$ terms and IMLL proof nets}
We can prove
the existence of a full and faithful embedding from 
the equivalence classes
of linear $\lambda$-terms up to $=_{\beta \eta {\rm c}}$
into the set of normal IMLL proof nets in the sense of \cite{Mat07}.
The proof is given in Appendix~\ref{appIMLLProofNets} with a brief introduction to IMLL proof nets.

\section{The Linear Distributive Transformation}
In this section we recall some definitions and results in \cite{Mat07}.
In \cite{Mat07}, most results are given by IMLL proof nets, not by the linear $\lambda$-calculus. 
But we have already given a full and faithful embedding from linear $\lambda$-terms to IMLL proof nets. 
So those results can be used for the linear $\lambda$-calculus freely. 
\begin{definition}
A linear $\lambda$-term ${\tt t}$ is implicational if 
there are neither ${\tt let}$ constructors nor $\mbox{\tt (} -\mbox{\tt ,} -\mbox{\tt )}$ constructors in ${\tt t}$. \\
A type ${\tt A}$ is implicational if there are no ${\tt A1 * A2}$ tensor subformulas in ${\tt A}$.
The order of an implicational formula ${\tt A}$, ${\rm order}({\tt A})$ is defined inductively as follows:
\begin{enumerate}
\item ${\tt A}$ is a propositional variable $\mbox{\tt 'a}$, then ${\rm order}({\tt A}) = 1$.
\item ${\tt A}$ is ${\tt A1} \mbox{\tt ->} \cdots \mbox{\tt ->} {\tt An} \mbox{\tt ->} \mbox{\tt 'a}$, then 
${\rm order}({\tt A})$ is
\[
\max \{ {\rm order}({\tt A1}), \ldots {\rm order}({\tt An}) \} \} + 1 
\]
\end{enumerate}
\end{definition}
The following proposition is the linear lambda calculus version of Corollary 2 in \cite{Mat07},
which says that
any different two terms of a type can be mapped 
into different two terms of another (but possibly the same) type with lower order (more precisely, less than $4$) without any tensor connectives injectively.
The purpose is to transform given terms into terms that can be treated easily. 
\begin{proposition}[Linear Distributive Transformation]
\label{propLDT-basic}
Let ${\tt A}$ be a type and ${\tt s1}$ and ${\tt s2}$ be two different closed terms of ${\tt A}$ up to $=_{\beta \eta {\rm c}}$.
Then there is a linear $\lambda$-term ${\tt LDTr\_A}$ such that 
${\tt LDTr\_A} \,  \, {\tt s1} \neq_{\beta \eta {\rm c}} {\tt LDTr\_A} \,  \, {\tt s2}$ 
and both ${\tt LDTr\_A} \,  \, {\tt s1}$ and ${\tt LDTr\_A} \,  \, {\tt s2}$ are a closed term of
an implicational type ${\tt A0}$ whose order is less than four. 
\end{proposition}
After obtaining two different closed terms
${\tt LDTr\_A} \, \, \mbox{\tt s1}$ and ${\tt LDTr\_A} \, \, \mbox{\tt s2}$
of the same implicational type $\mbox{\tt A0}$ with order less than four using the proposition,
we apply a term $\mbox {\tt s'}$ with poly-type $\mbox{\tt A0} \mbox{\tt ->} \mbox{\tt B}$, which is defined in the next section, and we obtain 
\[
\mbox{\tt s'} \, \, \mbox{\tt (} \mbox{\tt LDTr\_A} \, \, \mbox{\tt s1} \mbox{\tt )} =_{\beta \eta} \mbox{\tt t1}
\quad \quad
\mbox{and}
\quad \quad
\mbox{\tt s'} \, \, \mbox{\tt (} \mbox{\tt LDTr\_A} \, \, \mbox{\tt s2} \mbox{\tt )} =_{\beta \eta} \mbox{\tt t2}
\]
such that 
two closed terms $\mbox{\tt t1}$ and $\mbox{\tt t2}$ of type $\mbox{\tt B}$
are outputs of the intended specification.
This is an overview of our proof of Theorem~\ref{thmStrongTypedBohmTheorem}(Strong Typed B\"{o}hm Theorem).
In order to construct the term $\mbox {\tt s'}$, it is convenient to introduce a simple notion of model theory.
\begin{definition}[The Second-order Linear Term System]
(1) The language:
\begin{enumerate}
\item [(a)] A denumerable set of variables ${\rm Var}$: Elements of ${\rm Var}$ are denoted by $x_1,x_2, \ldots$. 
\item [(b)] A denumerable set of second-order variables ${\rm SVar}$: 
Elements of ${\rm SVar}$ are denoted by $G_1, G_2, \ldots$.
Each element of $G$ of ${\rm SVar}$ has its arity ${\rm arity}(G) \ge 1$. 
\end{enumerate}
(2) The set ${\rm SLT}$ of the terms of the language is defined inductively:
\begin{enumerate}
\item [(a)] If $x \in {\rm Var}$ then $x \in {\rm SLT}$.
\item [(b)] If $\{ t_1, \ldots, t_n \} \subseteq {\rm SLT}$, $G \in {\rm SVar}$ has arity $n$ 
and $t_i$ and $t_j$ have disjoint variables for each $i, j \, (i \neq j)$, then
$G(t_1, \ldots, t_n) \in {\rm SLT}$. 
\end{enumerate}
(3) Assignments: 
\begin{enumerate}
\item [(a)]
A variable assignment is a function $\rho_1 : {\rm Var} \to \{ 0, 1 \}$.
\item [(b)] A second-order variable assignment is a function $\rho_2$ from ${\rm SVar}$ to the set ${\rm CP}$, where 
${\rm CP}$ is the set of constant functions and (positive) projection functions on ${ \{ 0, 1 \} }^{n}$ into $\{ 0, 1 \}$ for each $n \ge 1$. 
\end{enumerate}
(4) Models:
A model for ${\rm SLT}$ ${ [| - |] }_{\langle \rho_1, \rho_2 \rangle} : {\rm SLT} \to \{ 0, 1 \}$ is determined uniquely for a given
$\langle \rho_1, \rho_2 \rangle$ as follows:
\begin{enumerate}
\item [(a)] ${ [| x |] }_{\langle \rho_1, \rho_2 \rangle} = \rho_1(x)$. 
\item [(b)] ${ [| G(t_1, \ldots, t_n) |] }_{\langle \rho_1, \rho_2 \rangle} = \rho_2(G)( {[| t_1|]}_{\langle \rho_1, \rho_2 \rangle}, \ldots, {[| t_n |]}_{\langle \rho_1, \rho_2 \rangle})$.
\end{enumerate}
\end{definition}
We note that in the definition above, to each second-order variable, a constant function or a (positive) projection is assigned. 
The following proposition is Proposition 25 in \cite{Mat07}.
\begin{proposition}
\label{propSLT-basic}
Let $s_1, s_2$ be in ${\rm SLT}$.
If $s_1 \neq s_2$ then there are a variable assignment $\rho_1$ and a second-order variable assignment $\rho_2$
such that ${[| s_1 |]}_{\langle \rho_1, \rho_2 \rangle} \neq {[| s_2 |]}_{\langle \rho_1, \rho_2 \rangle}$.
\end{proposition}
This proposition essentially uses linearity: 
for example we can not separate $f(x)$ and $f(f(f(x)))$ over $\{ 0, 1 \}$.
Then as observed in \cite{Mat07}, we note that 
an implicational closed term ${\tt s}$ of a type ${\tt A}$ whose order is less than 4 is identified with an element $s$ of ${\rm SLT}$.
So, without loss of generality, 
we can write ${\tt s}$ as a closed linear term 
\[
\mbox{\tt fn x1=>} \cdots \mbox{\tt =>} \mbox{\tt fn xn=>} \mbox{\tt fn G1=>} \cdots \mbox{\tt =>} \mbox{\tt fn Gm =>} \, \, \mbox{\tt s0} 
\]
where 
the principal type of ${\tt s}$ has the following form: 
\[
\begin{array}{l}
\overbrace{\mbox{\tt 'a}_{01} \mbox{\tt ->} \, \, \cdots \, \, \mbox{\tt ->} \mbox{\tt 'a}_{0n}}^{n}
\mbox{\tt ->} \\
\quad \quad 
\mbox{\tt (} \overbrace{\mbox{\tt 'a}_{11} \mbox{\tt ->} \cdots \mbox{\tt ->'a}_{1k_1}}^{k_1} \mbox{\tt ->} \mbox{\tt 'a}_{10} \mbox{\tt )}
\mbox{\tt ->}
\cdots 
\mbox{\tt ->}
\mbox{\tt (} \overbrace{\mbox{\tt 'a}_{m1} \mbox{\tt ->} \cdots \mbox{\tt ->} \mbox{'a}_{mk_m}}^{k_m} \mbox{\tt ->} \mbox{\tt 'a}_{m0} \mbox{\tt )} \\
\quad \quad \quad \quad
\mbox{\tt ->} \mbox{\tt 'a}_{00} \, \, . 
\end{array}
\]
and 
each positive (resp. negative) occurrence of $\mbox{\tt 'a}_{ij}$ in the type has
the corresponding exactly one negative (resp. positive) occurrence of $\mbox{\tt 'a}_{ij}$.
Unlike the weak typed B\"{o}hm theorem in \cite{Mat07},
each $\mbox{\tt 'a}_{ij}$ will not be instantiated with the same type 
in main theorems in this paper:
it may be instantiated with an implicational type with higher order.
For this reason we need the notion of {\it poly-types}, which will be introduced in the next section. 

\section{Poly-Types}
In this section we introduce the notion of poly-types, which is the key concept in this paper.
For that purpose we need to introduce some notions. 
\paragraph{Principal Type Theorem}
A {\it type substitution} is a function from type variables to types.
It is well-known that any type substitution is uniquely extended to a function from types to types. 
A type {\tt A} is an instance of a type {\tt B} if 
there is a type substitution $\theta$ such that $\mbox{\tt A} = \mbox{\tt B} \theta$. 
A type {\tt A} is a principal type of a linear term $t$ if 
(i) for some typing context $\Gamma$, 
$\Gamma \vdash \mbox{\tt t} : \mbox{\tt A}$ is derivable and
(ii) when $\Gamma' \vdash {\tt t} : {\tt A'}$ is derivable,
${\tt A'}$ and $\Gamma'$ are an instance of {\tt A} and $\Gamma$ respectively.
By the definition, if both {\tt A} and {\tt A'} are principal types of {\tt t}, then
{\tt A} is an instance of {\tt A'} and vice versa.
So we can call {\tt A} {\it the principal type} of ${\tt t}$ without ambiguity and write it as ${\rm PT}(\mbox{\tt t})$.
An untyped $\lambda$-term {\tt t} is defined by the following syntax:
\[
\mbox{\tt t} ::= \mbox{\tt x} \, \, | \, \, \mbox{\tt t} \, \mbox{\tt s} \, \, | \, \, \mbox{\tt fn x=>t} \, \, | \, \, \mbox{\tt (t} \, {\tt ,} \, \mbox{\tt s)} \, \, | \, \, \mbox{\tt let val (x,y)=s in t}
\]
An untyped linear $\lambda$-term {\tt t} is an untyped $\lambda$-term such that
each free or bound variable in {\tt t} occurs exactly once in {\tt t}. 
\begin{proposition}
If an untyped linear $\lambda$-term $\mbox{\tt t}$ is typable by the type assignment system in the previous section, then 
it has the principal type ${\rm PT}(\mbox{\tt t})$
\end{proposition}
\begin{proof}
  By assumption, we have a derivation for the term ${\tt t}$ with a type. 
  Then by applying an easily modified version of the main result of \cite{DM82} (see Section 7 of \cite{DM82}) augmented with the ${\tt \ast}$ connective to $\mbox{\tt t}$, we have a derivation for the term $\mbox{\tt t}$ with the principal type.
$\Box$
\end{proof}
Since our linear $\lambda$-calculus has the $\mbox{\tt let}$-constructor and the $\mbox{\tt (} -\mbox{\tt ,} -\mbox{\tt )}$ constructor, 
any untyped $\lambda$-term is not necessarily typable.
A counterexample is {\tt let val (x,y)=fn z=>z in (x, y)}.
If the system has neither the {\tt let}-constructor nor the $\mbox{\tt (} -\mbox{\tt ,} -\mbox{\tt )}$ constructor,
then any untyped $\lambda$-term is typable (see Theorem 4.1 of \cite{Hin89}). 
\paragraph{Poly-types}
\begin{example}
\label{exAp}
{\rm
The following two terms are the basic constructs in \cite{Mai04}:\\
{\tt - fun True x y z = z x y;} \\
{\tt - fun False x y z = z y x;} \\
The terms {\tt True} and {\tt False} can be considered as the two normal terms of
\[
\MBB{B}_{\rm HM} = \mbox{\tt 'a} \mbox{\tt ->} \mbox{\tt 'a} \mbox{\tt ->} \mbox{\tt (}\mbox{\tt 'a} \mbox{\tt ->} \mbox{\tt 'a} \mbox{\tt ->} \mbox{\tt 'a} \mbox{\tt )} \mbox{\tt ->} \mbox{\tt 'a} . \]
The following term can be considered as a {\it not} gate for $\MBB{B}_{\rm HM}$:\\
{\tt - fun Not\_POLY p = p False True (fn f=>fn g=>(erase\_3 g) f);}\\
where \\
{\tt - fun I x = x;} \\
{\tt - fun erase\_3 p = p I I I;}\\
We explain the reason in the following. 
The term $\mbox{\tt Not\_POLY}$ has 
types ${\tt A0}  \mbox{\tt ->} \MBB{B}_{\rm HM}$
and ${\tt A1}  \mbox{\tt ->} \MBB{B}_{\rm HM}$, where
\[
\begin{array}{ll}
{\tt A0}  =  {\tt X0} \mbox{\tt ->} {\tt Y0}  \mbox{\tt ->} ({\tt X0} \mbox{\tt ->} {\tt Y0} \mbox{\tt ->} {\tt Z0}) \mbox{\tt ->} {\tt Z0} & {\tt X0}  =  {\tt A} = {\tt Z0} \\
{\tt Y0}  =  {\tt P} \mbox{\tt ->} ({\tt A} \mbox{\tt ->} {\tt A}) \mbox{\tt ->} ({\tt P} \mbox{\tt ->} {\tt P}) \mbox{\tt ->} ({\tt A} \mbox{\tt ->} {\tt A}) & \\
{\tt P}  =  ({\tt A} \mbox{\tt ->} {\tt A}) \mbox{\tt ->} ({\tt A} \mbox{\tt ->} {\tt A})  &
{\tt A}  = \MBB{B}_{\rm HM}
\\
{\tt A1}  =  {\tt X1} \mbox{\tt ->} {\tt Y1} \mbox{\tt ->} ({\tt Y1} \mbox{\tt ->} {\tt X1} \mbox{\tt ->} {\tt Z1}) \mbox{\tt ->} {\tt Z1} & 
{\tt Y1}  =  {\tt A} = {\tt Z1} \\
{\tt X1}  =  ({\tt A} \mbox{\tt ->} {\tt A}) \mbox{\tt ->} {\tt P} \mbox{\tt ->} ({\tt P} \mbox{\tt ->} {\tt P}) \mbox{\tt ->} ({\tt A}  \mbox{\tt ->} {\tt A})
&
\end{array}
\]
Observe that ${\tt A0} \neq {\tt A1}$.
Moreover it is easy to see that there is no type substitution $\theta$ such that
$\theta( {\tt A0}) = \theta( {\tt A1})$.
On the other hand,
two terms $\mbox{\tt True}$ and $\mbox{\tt False}$ have 
the principal types
\[
\mbox{\tt 'a} \mbox{\tt ->} \mbox{\tt 'b}  \mbox{\tt ->} (\mbox{\tt 'a} \mbox{\tt ->} \mbox{\tt 'b} \mbox{\tt ->} \mbox{\tt 'c}) \mbox{\tt ->} \mbox{\tt 'c}
\quad 
\mbox{and}
\quad
\mbox{\tt 'a} \mbox{\tt ->} \mbox{\tt 'b}  \mbox{\tt ->} (\mbox{\tt 'b} \mbox{\tt ->} \mbox{\tt 'a} \mbox{\tt ->} \mbox{\tt 'c}) \mbox{\tt ->} \mbox{\tt 'c},
\]
respectively.
Moreover, these types have instances ${\tt A0}$ and ${\tt A1}$ respectively. 
As a result,
two application terms $\mbox{\tt Not\_POLY} \, \, \, \mbox{\tt True}$
and 
$\mbox{\tt Not\_POLY} \, \, \, \mbox{\tt False}$
have a type $\MBB{B}_{\rm HM}$. 
}
\end{example}
Example~\ref{exAp} motivates the following definition.
\begin{definition}
\label{def-PolyTypability}
Let $t$ and $s$ be two closed linear $\lambda$-terms 
such that 
$\vdash \mbox{\tt t:A'->B'}$ and 
$\vdash \mbox{\tt s:A}$ are derivable
and for some type substitution $\theta_0$, 
$\theta_0({\rm PT}({\tt t})) = \mbox{\tt A0->B}$
and 
$\theta_0({\rm PT}({\tt s})) = \mbox{\tt A0}$.
Then we say that 
the term $\mbox{\tt t}$ is  poly-typable by $\mbox{\tt A->B}$ w.r.t. $\mbox{\tt s}$.
\end{definition}
When $\mbox{\tt t}$ is  poly-typable by $\mbox{\tt A->B}$ w.r.t. $\mbox{\tt s}$, 
observe that $\vdash \mbox{\tt t:A->B}$ is not necessarily derivable. 
For example, the term $\mbox{\tt Not\_POLY}$ is not typable by $\MBB{B}_{\rm HM} \mbox{\tt ->} \MBB{B}_{\rm HM}$,
but is  poly-typable by $\MBB{B}_{\rm HM} \mbox{\tt ->} \MBB{B}_{\rm HM}$ 
w.r.t. $\mbox{\tt True}$ and $\mbox{\tt False}$ respectively.
But then note that $\mbox{\tt t} \, \mbox{\tt s}$ has type $\mbox{\tt B}$ in the usual sense.
For example,  both $\mbox{\tt Not\_POLY} \, \, \, \mbox{\tt True}$ and $\mbox{\tt Not\_POLY} \, \, \, \mbox{\tt False}$
have type $\MBB{B}_{\rm HM}$. 

The importance of Definition~\ref{def-PolyTypability} is the composability of two poly-typable terms.
The proof of the following proposition is easy.
\begin{proposition}
\label{prop-PolyTypability}
Let $\mbox{\tt t}$ be poly-typable by $\mbox{\tt A->B}$ w.r.t. 
two terms $\mbox{\tt s}$ and $\mbox{\tt s'}$ with type ${\tt A}$.
Moreover let $\mbox{\tt t'}$ be poly-typable by $\mbox{\tt B->C}$ w.r.t.
the two terms $\mbox{\tt t s}$ and $\mbox{\tt t s'}$. 
Then the term $\mbox{\tt fn x=>(t'(t x))}$ are poly-typable by $\mbox{\tt A->C}$ w.r.t $\mbox{\tt s}$ and $\mbox{\tt s'}$.
\end{proposition}
We need a generalization of the definition above. 
Let $\mbox{\tt t}$ and $\mbox{\tt s}_i \, (1 \le i \le n)$ be closed linear $\lambda$-terms 
such that 
$\vdash \mbox{\tt t:A'}_{1} \mbox{\tt ->} \cdots \mbox{\tt ->} \mbox{\tt A'}_{n} \mbox{\tt ->} \mbox{\tt B'}$ and 
$\vdash \mbox{\tt s}_{i} \mbox{\tt :} \mbox{\tt A}_{i}$ are derivable.
If for some type substitution $\theta$, 
we have $\theta({\rm PT}({\tt t})) = \mbox{\tt A''}_{1} \mbox{\tt ->} \cdots \mbox{\tt ->} \mbox{\tt A''}_{n} \mbox{\tt ->} \mbox{\tt B}$
and $\theta({\rm PT}({\tt s}_{i})) = \mbox{\tt A''}_{i}$, then 
we say that 
the term $\mbox{\tt t}$ is poly-typable by $\mbox{\tt t:A}_{1} \mbox{\tt ->} \cdots \mbox{\tt ->} \mbox{\tt A}_{n} \mbox{\tt ->} \mbox{\tt B}$ w.r.t. $\mbox{\tt s}_{i}$.
\begin{remark}
Poly-types are used in \cite{Mai04} without referring
to it explicitly.
Let {\tt A} be
a uniform data type consisting of exactly one type variable {\tt 'a}
(for example, 
$\MBB{B}_{\rm HM} = \mbox{\tt 'a} \mbox{\tt ->} \mbox{\tt 'a} \mbox{\tt ->} \mbox{\tt (}\mbox{\tt 'a} \mbox{\tt ->} \mbox{\tt 'a} \mbox{\tt ->} \mbox{\tt 'a} \mbox{\tt )} \mbox{\tt ->} \mbox{\tt 'a}$).
In general,
the principal type of a closed term of {\tt A}
is more general than {\tt A}.
The basic idea is to utilize the difference ingeniously.
By using more general types, we can acquire more expressive power. 
\end{remark}

\section{Strong Typed B\"{o}hm Theorem}
In this section we prove the first main theorem of this paper: a version of {\it the typed B\"{o}hm theorem} with regard to $=_{\beta \eta {\rm c}}$.
First we give some preliminary results, which state
that
for any types $\mbox{\tt A}$ and $\mbox{\tt B}$ having at least one closed term,
we can always represent any projection from $\mbox{\tt A} \times \cdots \times \mbox{\tt A}$ to $\mbox{\tt A}$
and any constant function from $\mbox{\tt A}$ to $\mbox{\tt B}$
using the notion of poly-types. 
\begin{lemma}[Projection Lemma]
\label{lemmaProjectionLemma}
Let ${\tt A}$ be a type having at least one closed term.
For any type ${\tt B}$, 
there is a closed term ${\tt t}$ that is poly-typable by $\mbox{\tt A} \mbox{\tt ->} (\mbox{\tt B} \mbox{\tt ->} \mbox{\tt B})$ 
w.r.t. any closed term $\mbox{\tt s}$ of $\mbox{\tt A}$
such that 
\[
{\tt t} \, \, {\tt s} =_{\beta \eta {\rm c}} \mbox{\tt I}
\]
\end{lemma}
\begin{proof}
The term ${\tt t}$ that we are looking for has the following form:
\[
\mbox{\tt fun t x0 = LDTr\_A} \, \, \mbox{\tt x0} \, \overbrace{\mbox{\tt I} \, \, \cdots \, \, \mbox{\tt I}}^{n} \, \, \overbrace{\mbox{\tt u1} \, \, \cdots \, \, \mbox{\tt um}}^{m} \mbox{\tt ;}
\]
where $\mbox{\tt LDTr\_A}$ is the closed term obtained using Proposition~\ref{propLDT-basic}
and the closed term ${\tt uj}$ is defined by 
\[
\mbox{\tt fun uj x1} \cdots {\mbox{\tt xkj-1 xkj}} \, \, \mbox{\tt = x1 (} \, \, \cdots \, \, \mbox{\tt (xkj-1 (xkj I))} \cdots \mbox{\tt )} \mbox{\tt ;}
\]
for each $j \, (1 \le j \le m)$.
We note that the only occurrence of ${\tt I}$ in ${\tt uj}$ is typed by $\mbox{\tt 'a->'a}$ in the principal typing,
which implies that it can be typed by $\mbox{\tt B->B}$.
We also observe that the principal type of $\mbox{\tt LDTr\_A} \, \, \mbox{\tt s0}$ has the following form: 
\[
\begin{array}{l}
\overbrace{\mbox{\tt 'a}_{01} \mbox{\tt ->} \, \, \cdots \, \, \mbox{\tt ->} \mbox{\tt 'a}_{0n}}^{n}
\mbox{\tt ->} \\
\quad \quad 
\mbox{\tt (} \overbrace{\mbox{\tt 'a}_{11} \mbox{\tt ->} \cdots \mbox{\tt ->'a}_{1k_1}}^{k_1} \mbox{\tt ->} \mbox{\tt 'a}_{10} \mbox{\tt )}
\mbox{\tt ->}
\cdots 
\mbox{\tt ->}
\mbox{\tt (} \overbrace{\mbox{\tt 'a}_{m1} \mbox{\tt ->} \cdots \mbox{\tt ->} \mbox{'a}_{mk_m}}^{k_m} \mbox{\tt ->} \mbox{\tt 'a}_{m0} \mbox{\tt )} \\
\quad \quad \quad \quad
\mbox{\tt ->} \mbox{\tt 'a}_{00} \, \, . 
\end{array}
\]
where
each positive (resp. negative) occurrence of $\mbox{\tt 'a}_{ij}$ in the type has
the corresponding exactly one negative (resp. positive) occurrence of $\mbox{\tt 'a}_{ij}$.
Since the combinator ${\tt I}$ is substituted for each bounded variables ${\tt xi} \, (1 \le i \le k_j)$ in ${\tt uj}$, 
the application term  ${\tt (} {\tt t} \, {\tt  s} {\tt )}$ is reduced to ${\tt I}$.
Since the only occurrence of ${\tt I}$ in ${\tt uj}$ can be typed by $\mbox{\tt B->B}$,
the term  ${\tt (} {\tt t} \, {\tt  s} {\tt )}$ can be typed by $\mbox{\tt B->B}$.
This means that
${\tt t}$ can be poly-typed by $\mbox{\tt A->(B->B)}$ w.r.t. any closed term of type $\mbox{\tt A}$.
$\Box$
\end{proof}
Note that a type variable $\mbox{\tt 'a}_{ij}$ may be instantiated with an implicational type of very higher order
in the term {\tt t}.
For this reason we need the notion of {\it poly-types}.

The following corollary, which is a generalization of the proposition above to $n$-ary case, is obtained as a direct consequence of it. 
\begin{corollary}
\label{corProjectionLemma-1}
Let ${\tt A}$ be a type having at least one closed term.
There is an $i$-th projection that is poly-typable by
$\overbrace{\mbox{\tt A->} \, \, \cdots \, \, \mbox{\tt ->A}}^{n} \mbox{\tt ->A}$
for each $i \, (1 \le i \le n)$ and 
for any $n$.
\end{corollary}
\begin{proof}
Think $\overbrace{\mbox{\tt A->} \, \, \cdots \, \, \mbox{\tt ->A}}^{n} \mbox{\tt ->A}$ 
as $\overbrace{\mbox{\tt A->} \, \, \cdots \, \, \mbox{\tt ->A}}^{n-1} \mbox{\tt ->} \mbox{\tt(} \mbox{\tt A->A} \mbox{\tt)}$.
Then let $\mbox{\tt bxi}$ be
\[
\mbox{\tt LDTr\_A} \, \mbox{\tt xi} \, \overbrace{\mbox{\tt I} \, \, \cdots \, \, \mbox{\tt I}}^{n_i} \, \, \overbrace{\mbox{\tt u1} \, \, \cdots \, \, \mbox{\tt umi}}^{m_i} \mbox{\tt ;}
\]
for $i \, (0 \le i \le n-1)$.
The term ${\tt t}$ that we are looking for has the following form:
\[
\mbox{\tt fun t x0} \cdots \mbox{\tt xn-1} \, \, \mbox{\tt xn}
\mbox{\tt = bx0 (} \, \, \cdots \, \, \mbox{\tt (bxn-2 (bxn-1 xn))} \cdots \mbox{\tt )} \mbox{\tt ;}
\]
$\Box$
\end{proof}
\begin{lemma}[Constant Function Lemma]
\label{lemmaConstantFunctionLemma}
Let ${\tt A}$ and ${\tt B}$ be types having at least one closed term.
Let ${\tt u}$ be a closed term of ${\tt B}$.
Then there is a closed term ${\tt t}$ that is poly-typable by $\mbox{\tt A} \mbox{\tt ->} \mbox{\tt B}$ 
w.r.t. any closed term $\mbox{\tt s}$ of $\mbox{\tt A}$
such that 
\[
{\tt t} \, \, {\tt s} =_{\beta \eta {\rm c}} \mbox{\tt u}
\]
\end{lemma}
\begin{proof}
  Let $\mbox{\tt proj}$ be the term which is 
  poly-typable by $\mbox{\tt A} \mbox{\tt ->} (\mbox{\tt B} \mbox{\tt ->} \mbox{\tt B})$ 
  w.r.t. any closed term $\mbox{\tt s}$ of $\mbox{\tt A}$
  obtained using Lemma~\ref{lemmaProjectionLemma}.
  The term ${\tt t}$ that we are looking for is the following term:
\[
\mbox{\tt fun t x0  = proj x0 u}
\]
$\Box$
\end{proof}
\begin{corollary}
\label{corConstantFunctionLemma}
Let ${\tt A}$ be a type having at least one closed term.
Let ${\tt s}$ be such a closed term. 
There is a constant function that always returns $\mbox{\tt s}$ and is poly-typable by
$\overbrace{\mbox{\tt A->} \, \, \cdots \, \, \mbox{\tt ->A}}^{n} \mbox{\tt ->A}$ 
for any $n$.
\end{corollary}
\begin{theorem}[Strong Typed B\"{o}hm Theorem]
\label{thmStrongTypedBohmTheorem}
For any types ${\tt A}$ and ${\tt B}$,
when any two different closed terms ${\tt s1}$ and ${\tt s2}$ of type ${\tt A}$
and any closed terms ${\tt u1}$ and ${\tt u2}$ of type ${\tt B}$ are given, 
there is a closed term ${\tt t}$ that is poly-typable by $\mbox{\tt A->B}$ such that
\[
\mbox{\tt t} \, \, \mbox{\tt s1} =_{\beta \eta {\rm c}} \mbox{\tt u1}
\, \, \mbox{and} \, \, 
\mbox{\tt t} \, \, \mbox{\tt s2} =_{\beta \eta {\rm c}} \mbox{\tt u2}
\]
\end{theorem}
\begin{proof}
The term ${\tt t}$ that we are looking for has the following form:
\[
\mbox{\tt fun t x0 = LDTr\_A} \, \, \mbox{\tt x0} \, \overbrace{\mbox{\tt v1} \, \, \cdots \, \, \mbox{\tt vn}}^{n} \, \, \overbrace{\mbox{\tt w1} \, \, \cdots \, \, \mbox{\tt wm}}^{m} \mbox{\tt ;}
\]
By Proposition~\ref{propLDT-basic}, we have $\mbox{\tt LDTr\_A} \, \, \mbox{\tt s1} \neq_{\beta \eta {\rm c}} \mbox{\tt LDTr\_A} \, \, \mbox{\tt s2}$.
Then since $\mbox{\tt LDTr\_A} \, \, \mbox{\tt s1}$ and $\mbox{\tt LDTr\_A} \, \, \mbox{\tt s2}$ are typable by a common type
with order less than four, as observed before, 
they are identified with terms $s_1$ and $s_2$ in ${\rm SLT}$ respectively such that $s_1 \neq s_2$. 
Then by Proposition~\ref{propSLT-basic}, 
there are a variable assignment $\rho_1$ and a second-order variable assignment $\rho_2$
such that ${[| s_1 |]}_{\langle \rho_1, \rho_2 \rangle} \neq {[| s_2 |]}_{\langle \rho_1, \rho_2 \rangle}$.
Then following $\rho_1$, we choose ${\tt u1}$ or ${\tt u2}$ as the subterm $\mbox{\tt vi}$ (with type ${\tt B}$) of ${\tt t}$ for each $i \, (1 \le i \le n)$ 
and following $\rho_2$, we choose a constant function or a projection as the subterm $\mbox{\tt wj}$   
(with poly-type $\overbrace{\mbox{\tt B->} \cdots \mbox{\tt ->B}}^{k_j} \mbox{\tt ->B}$) of $\mbox{\tt t}$ for each $j \, (1 \le j \le m)$.
These constant functions and projections are obtained using Projection and Constant Function Lemmas.
Note that these constant functions and projections can be composed by Proposition~\ref{prop-PolyTypability}
such that the closed term {\tt t} is poly-typable appropriately. 
It is obvious that the term ${\tt t}$ has the desired properties. 
$\Box$
\end{proof}
\begin{remark}
  Theorem~\ref{thmStrongTypedBohmTheorem} can be considered as a strong version of Corollary 6 in \cite{Mat07}.
  While Corollary 6 in \cite{Mat07} uses only uniform type instantiation,
  Theorem~\ref{thmStrongTypedBohmTheorem} uses poly-types.
  We can not prove Theorem~\ref{thmStrongTypedBohmTheorem} using only uniform type instantiation. 
  Appendix~\ref{secWhyNeedPoly-Types} gives a discussion of this matter. 
\end{remark}
\begin{corollary}
\label{corCopyTerm}
Let ${\tt s1}$ and ${\tt s2}$ be two closed terms of ${\tt A}$.
Then there is a closed term ${\tt Copy\_A\_n}$
such that 
\[
\mbox{\tt Copy\_A\_n} \, \, \mbox{\tt s1} =_{\beta \eta {\rm c}} \mbox{\tt (s1,} \cdots \mbox{\tt ,s1)}
\quad 
\mbox{\tt Copy\_A\_n} \, \, \mbox{\tt s2} =_{\beta \eta {\rm c}} \mbox{\tt (s2,} \cdots \mbox{\tt ,s2)}
\]
where 
\mbox{\tt s1} and \mbox{\tt s2}
occur in 
\mbox{\tt (s1,} $\cdots$ \mbox{\tt ,s1)}
and 
\mbox{\tt (s2,} $\cdots$ \mbox{\tt ,s2)}
$n$ times respectively. 
\end{corollary}
\begin{proof}
In Theorem~\ref{thmStrongTypedBohmTheorem}, 
one chooses $\overbrace{\mbox{\tt A} \mbox{\tt *} \cdots \mbox{\tt *} \mbox{\tt A}}^{n}$ as $\mbox{\tt B}$, 
and then 
\mbox{\tt (s1,} $\cdots$ \mbox{\tt ,s1)}
and 
\mbox{\tt (s2,} $\cdots$ \mbox{\tt ,s2)}
as 
\mbox{\tt u1}
and 
\mbox{\tt u2}
respectively.
$\Box$
\end{proof}
The next theorem claims that
in a limited situation
we can obtain a closed term representing a function from closed terms of a type 
to closed terms that may not be typable by the same implicational type, but are {\it poly-typable} by the type. 
\begin{theorem}[Poly-type Version of Strong Typed B\"{o}hm Theorem]
\label{thmPolytimeVersionOfStrongTypedBohmTheorem}
Let ${\tt s1}$ and ${\tt s2}$ denote two different closed terms with type ${\tt A}$, 
and ${\tt u1}$ and ${\tt u2}$ denote two different closed terms which are poly-typable by
$\mbox{\tt A0->B}$ w.r.t. two closed terms ${\tt r1}$ and ${\tt r2}$ with type ${\tt A0}$
such that 
$\{ \mbox{\tt u1} \, \, \mbox{\tt r1}, \, \mbox{\tt u1} \, \, \mbox{\tt r2}, \, \mbox{\tt u2} \, \, \mbox{\tt r1}, \, \mbox{\tt u2} \, \, \mbox{\tt r2}  \}$ is
a set of one or two closed terms (with type $B$). 
Then there is a closed term ${\tt t}$ that is poly-typable by $\mbox{\tt A->A0->B}$ such that
\[
\mbox{\tt t} \, \, \mbox{\tt s1} \, \, \mbox{\tt ri} =_{\beta \eta {\rm c}} \mbox{\tt u1} \, \, \mbox{\tt ri}
\, \, \mbox{and} \, \, 
\mbox{\tt t} \, \, \mbox{\tt s2} \, \, \mbox{\tt ri} =_{\beta \eta {\rm c}} \mbox{\tt u2} \, \, \mbox{\tt ri}
\]
for each $i \in \{ 1, 2 \}$. 
\end{theorem}
\begin{proof}
By Proposition~\ref{propLDT-basic}
there is a linear $\lambda$-term ${\tt LDTr\_A}$ such that 
${\tt LDTr\_A} \,  \, {\tt s1} \neq_{\beta \eta {\rm c}} {\tt LDTr\_A} \,  \, {\tt s2}$ 
and these terms can be regarded as different linearly labeled trees $T_1$ and $T_2$ respectively.
In the rest of the proof,
we assign a poly-typable first-order function to each leaf (which represented a first order variable in our proof of Theorem~\ref{thmStrongTypedBohmTheorem})
and a poly-typable first-order or second-order function to each internal node (which represented a second order variable in our proof of Theorem~\ref{thmStrongTypedBohmTheorem})
in $T_1$ and $T_2$, 
following the structure of trees $T_1$ and $T_2$.
The purpose is to construct a closed term {\tt t} such that
each of $\mbox{\tt t} \, \, \mbox{\tt s1}$ and $\mbox{\tt t} \, \, \mbox{\tt s2}$
represents a one argument boolean function satisfying the specification of the theorem.
The main tools are
Projection and Constant Function Lemmas and the Strong Typed B\"{o}hm Theorem.
We have two cases according to the structure of $T_1$ and $T_2$.
\begin{itemize}
\item The case where
both $T_1$ an $T_2$ have an $n$-ary second order variable $F$ $\, (n \ge 2)$ and a first or second order variable $G$ 
such that 
$G$ is above $F$ in both $T_1$ and $T_2$ and 
the position of $G$ in $T_1$ is different from that of $T_2$: \\
Furthermore, the case is divided into three cases. 
We assume that we choose $F$ to be the nearest one to $G$ in $T_1$ and 
the variable in $T_2$ that has the same position as $G$ in $T_1$ is $H$. 
\begin{itemize}
\item The case where there is a path from the root to a leaf, including $G$ in $T_1$ 
  such that the path does not include $H$,
  and
  when we interchange $G$ and $T_1$ with $H$ and $T_2$ respectively, the same thing happens: 
  \\
Without loss of generality, this case can be shown as Figure~\ref{fig-Memo-Figure-1}.
The term ${\tt t}$ that we are looking for has the following form:
\[
\begin{array}{l}
\mbox{\tt fun t x0 y0 = } \\
\quad \quad \mbox{\tt let val } \mbox{\tt (x1,} \, \, \cdots \, \, \mbox{\tt ,xn)} = \mbox{\tt Copy\_A0\_n}  \, \, \mbox{\tt y0} \, \, \mbox{\tt in} \\
\quad \quad \quad \quad \mbox{\tt LDTr\_A} \, \, \mbox{\tt x0} \, \overbrace{\mbox{\tt (v1 x1)} \, \, \cdots \, \, \mbox{\tt (vn xn)}}^{n} \, \, \overbrace{\mbox{\tt w1} \, \, \cdots \, \, \mbox{\tt wm}}^{m} \, \, \mbox{\tt end;}
\end{array}
\]
where 
the subterm $\mbox{\tt vi}$ that is poly-typable by $\mbox{\tt A0} \mbox{\tt ->} \mbox{\tt B}$
is obtained using Theorem~\ref{thmStrongTypedBohmTheorem}, representing
a surjection from $\{ \mbox{\tt r1}, \mbox{\tt r2} \}$ to 
one or two element set
$\{ \mbox{\tt u1} \, \, \mbox{\tt r1}, \, \mbox{\tt u1} \, \, \mbox{\tt r2}, \, \mbox{\tt u2} \, \, \mbox{\tt r1}, \, \mbox{\tt u2} \, \, \mbox{\tt r2}  \}$ 
for each $i \, (1 \le i \le n)$.
The subterm $\mbox{\tt wj}$ 
that is poly-typable by 
$\overbrace{\mbox{\tt B->} \cdots \mbox{\tt ->B}}^{k_j} \mbox{\tt ->B}$ for each $j \, (1 \le j \le m)$
is constructed 
from Projection Lemma w.r.t. an appropriate position 
except for $G$ and $H$.
For example the first argument projection is assigned to $F$ in Figure~\ref{fig-Memo-Figure-1}.
Then $G$ and $H$ are constructed in the following two steps: 
\begin{enumerate}
\item 
  First we construct terms $\mbox{\tt mj}$ with type $\mbox{\tt B->B}$
  using the Strong Typed B\"{o}hm Theorem (Theorem~\ref{thmStrongTypedBohmTheorem}). 
  The functions for $G$ and $H$ are
  the constant, identity, or negation functions, 
  depending on $\mbox{\tt u1}$ and $\mbox{\tt u2}$.
  Note that in order to represent the negation function we need the Strong Typed B\"{o}hm Theorem. 
\item 
Second from using $\mbox{\tt mj}$, we construct $\mbox{\tt wj}$ using Constant Function Lemma in order to discard
the unnecessary arguments.
The terms corresponding to $G$ and $H$ in Figure~\ref{fig-Memo-Figure-1} discard
the second argument. 
\end{enumerate}
\begin{figure}[htbp]
 \begin{center}
  \includegraphics[scale=0.5]{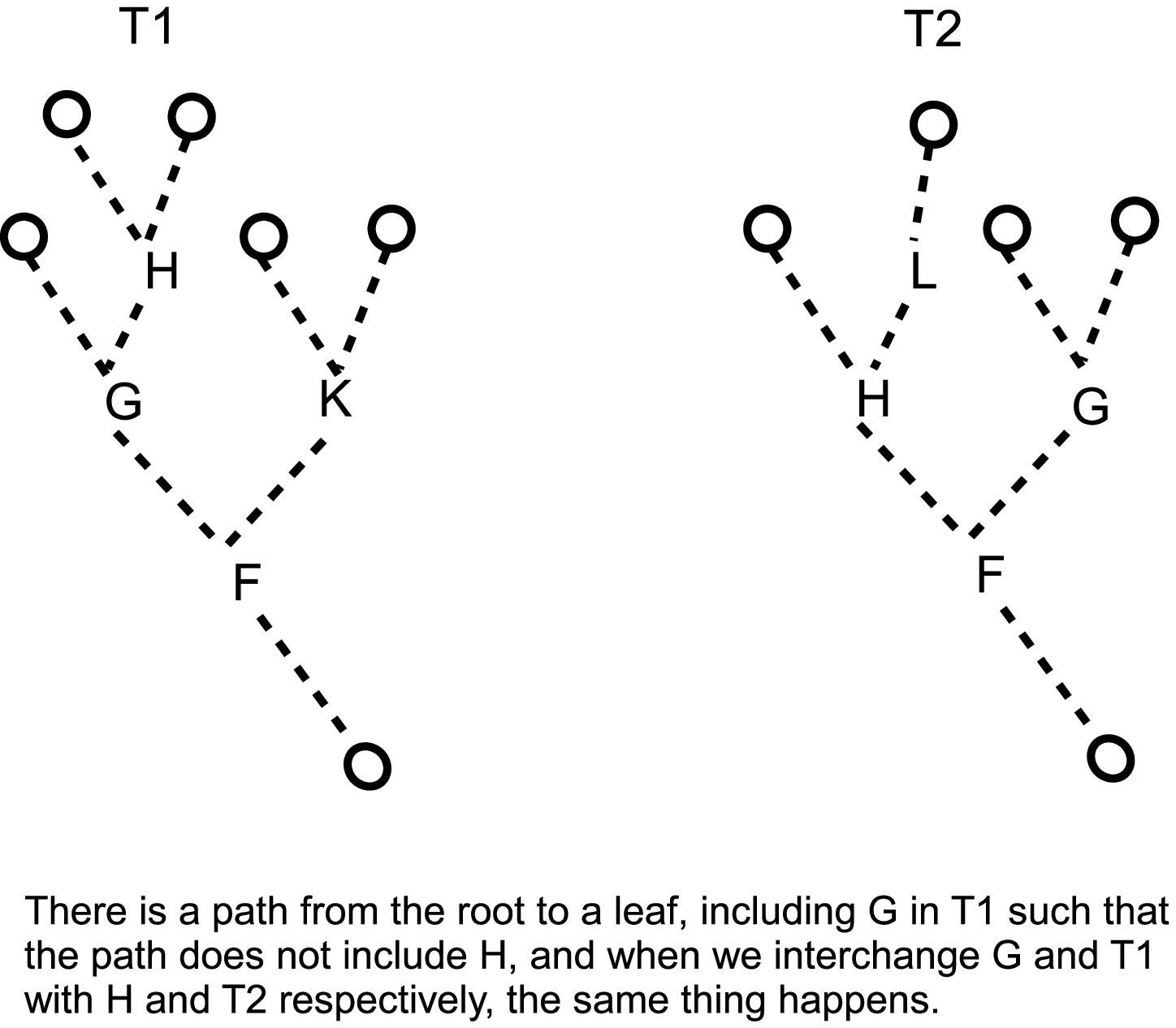}
 \end{center}
 \caption{Two different linearly labeled trees (1)}
 \label{fig-Memo-Figure-1}
\end{figure}
\item The case where (i) there is no any path from the root to a leaf, including $G$ in $T_1$ such that the path does not include $H$ and (ii) there is a path from the root to a leaf, including $H$ in $T_2$ such that the path does not include $G$: \\
  We assume that 
  the variable in $T_1$ that has the same position as $G$ in $T_2$ is $K$. 
  In this case, the following additional properties hold: 
  \begin{itemize}
  \item [(iii)] There is no any path from the root to a leaf, including $G$ in $T_2$ such that the path does not include $K$.
  \item [(iv)] there is a path from the root to a leaf, including $K$ in $T_1$ such that the path does not include $G$.
  \end{itemize}
  Otherwise, we can apply the immediately above case (replace $G$ and $H$ by $K$ and $G$ respectively). 
  In the case, $T_1$ and $T_2$ have the form of Figure~\ref{fig-Memo-Figure-2} or Figure~\ref{fig-Memo-Figure-3}
  without loss of generality.
  First we consider the case of Figure~\ref{fig-Memo-Figure-2}.
The term ${\tt t}$ that we are looking for has the following form:
\[
\begin{array}{l}
\mbox{\tt fun t x0 y0 = } \\
\quad \quad \mbox{\tt let val } \mbox{\tt (x1,} \, \, \cdots \, \, \mbox{\tt ,xn)} = \mbox{\tt Copy\_A0\_n}  \, \, \mbox{\tt y0} \, \, \mbox{\tt in} \\
\quad \quad \quad \quad \mbox{\tt LDTr\_A} \, \, \mbox{\tt x0} \, \overbrace{\mbox{\tt (v1 x1)} \, \, \cdots \, \, \mbox{\tt (vn xn)}}^{n} \, \, \overbrace{\mbox{\tt w1} \, \, \cdots \, \, \mbox{\tt wm}}^{m} \, \, \mbox{\tt end;}
\end{array}
\]
where 
the subterm $\mbox{\tt vi}$ that is poly-typable by $\mbox{\tt A0} \mbox{\tt ->} \mbox{\tt B}$
is
obtained using the Strong Typed B\"{o}hm Theorem (Theorem~\ref{thmStrongTypedBohmTheorem})
for each $i \, (1 \le i \le n)$,
representing a surjection from $\{ \mbox{\tt r1}, \mbox{\tt r2} \}$ to
one or two element set
$\{ \mbox{\tt u1} \, \, \mbox{\tt r1}, \, \mbox{\tt u1} \, \, \mbox{\tt r2}, \, \mbox{\tt u2} \, \, \mbox{\tt r1}, \, \mbox{\tt u2} \, \, \mbox{\tt r2}  \}$ 
and 
the subterm $\mbox{\tt wj}$ 
is poly-typable by
$\overbrace{\mbox{\tt B->} \cdots \mbox{\tt ->B}}^{k_j} \mbox{\tt ->B}$ for each $j \, (1 \le j \le m)$
obtained from Projection Lemma 
except that four terms assigned to $F$, $G$, $H$, and $K$ are selected according to the table immediately below
(and then Constant Function Lemma is applied in order to discard the unnecessary arguments):
\[
\begin{array}{|l|l|l|l|l|l|}
\hline 
u_1 & u_2 & \mbox{argument} & G & H & K \\ 
    &     & \mbox{choice of} \, \, F & & & \\ 
\hline \hline
\mbox{const.} & \mbox{const.} & \mbox{left} & \mbox{const.} & \mbox{const.} & \mbox{don't care} \\
\hline
\mbox{const.} & \mbox{id.} & \mbox{left} & \mbox{const.} & \mbox{id.} & \mbox{don't care} \\
\hline
\mbox{const.} & \mbox{neg.} & \mbox{left} & \mbox{const.} & \mbox{neg.} & \mbox{don't care} \\
\hline
\mbox{id.} & \mbox{const.} & \mbox{right} & \mbox{const.} & \mbox{don't care} & \mbox{id.} \\
\hline
\mbox{neg.} & \mbox{const} & \mbox{right} & \mbox{const.} & \mbox{don't care} & \mbox{neg.} \\
\hline
\mbox{id.} & \mbox{id.} & \mbox{left} & \mbox{id.} & \mbox{id} & \mbox{id.} \\
\hline
\mbox{neg.} & \mbox{neg.} & \mbox{left} & \mbox{id.} & \mbox{neg.} & \mbox{dont' care} \\
\hline
\mbox{id.} & \mbox{neg.} & \mbox{left} & \mbox{neg.} & \mbox{neg.} & \mbox{dont' care} \\
\hline
\mbox{neg.} & \mbox{id.} & \mbox{left} & \mbox{neg.} & \mbox{id.} & \mbox{dont' care} \\
\hline
\end{array}
\]
where
id., neg., and const. mean
the identity, negation, and constant functions respectively.
The term ``don't care'' means
that we can choose any one argument function for that place.

In the case of Figure~\ref{fig-Memo-Figure-3}, 
the form of the term ${\tt t}$ is the same as Figure~\ref{fig-Memo-Figure-2}. 
The only difference is that we assign one argument functions to the subterms $\mbox{\tt vi}$s corresponding to $x$ and $y$, according to the instructions for $H$ and $K$ in the above table respectively.
We can do the assignment using Theorem~\ref{thmStrongTypedBohmTheorem}.

\begin{figure}[htbp]
 \begin{center}
  \includegraphics[scale=0.5]{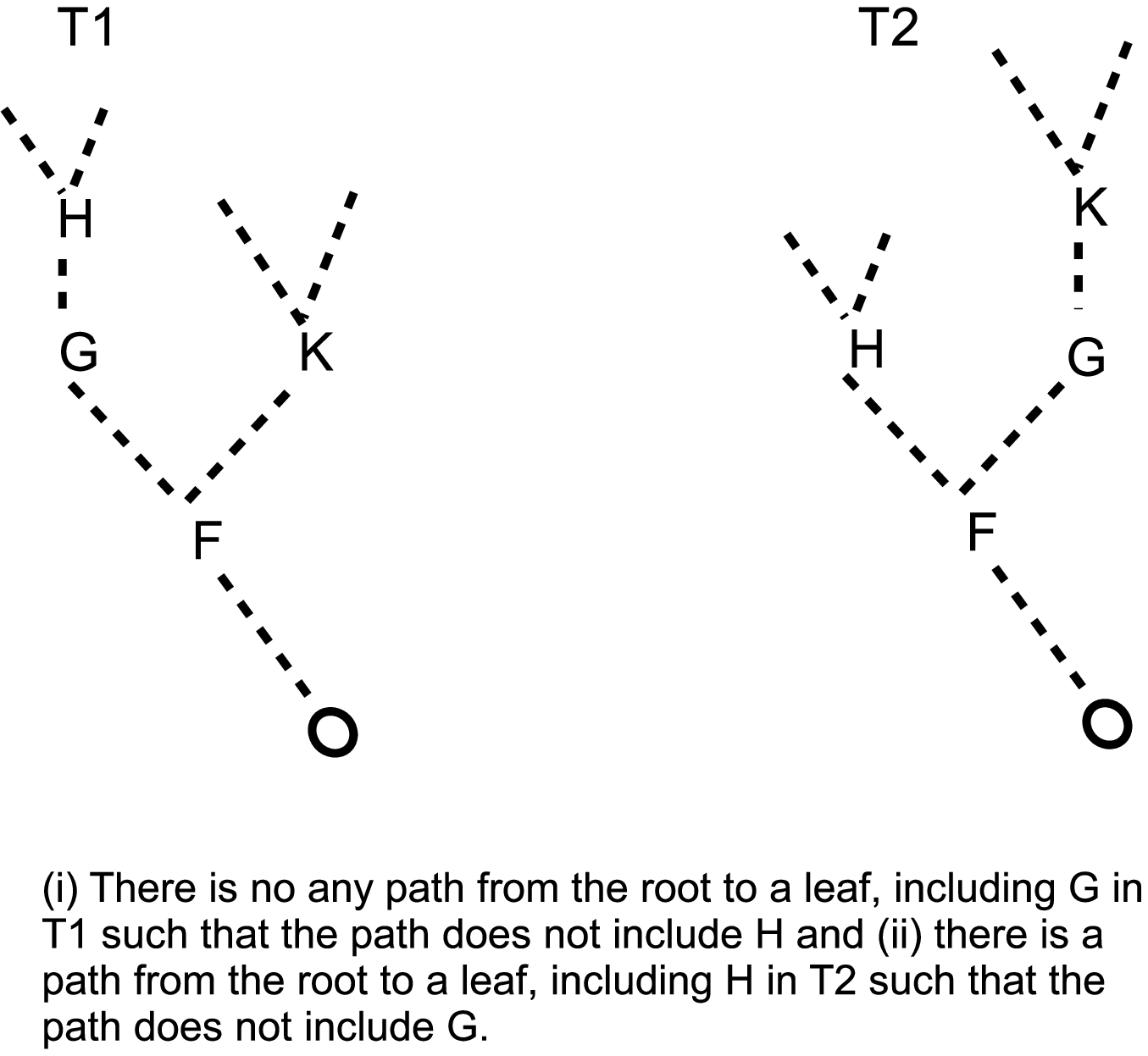}
 \end{center}
 \caption{Two different linearly labeled trees (2)}
 \label{fig-Memo-Figure-2}
\end{figure}
\begin{figure}[htbp]
 \begin{center}
  \includegraphics[scale=0.5]{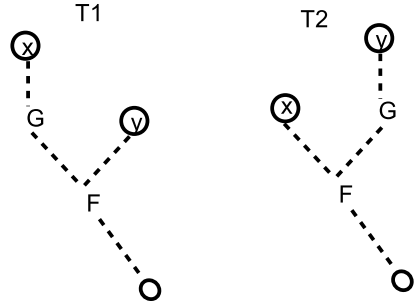}
 \end{center}
 \caption{Two different linearly labeled trees (3)}
 \label{fig-Memo-Figure-3}
\end{figure}
\item Otherwise: \\
  In this case,
  any path from the root to a leaf including $G$ (resp. $H$) in $T_1$ (resp. $T_2$) includes $H$ (resp. $G$) 
  above $G$ (resp. $H$).
  Without loss of generality, this case can be shown as Figure~\ref{fig-Memo-Figure-4}.
  The term ${\tt t}$ that we are looking for has the following form:
\[
\begin{array}{l}
\mbox{\tt fun t x0 y0 = } \\
\quad \quad \quad \quad \mbox{\tt LDTr\_A} \, \, \mbox{\tt x0} \, \overbrace{\mbox{\tt v1} \, \, \cdots \, \, \mbox{\tt vn}}^{n} \, \, \overbrace{\mbox{\tt w1} \, \, \cdots \, \, \mbox{\tt wm}}^{m} \, \, \mbox{\tt (t0 y0);}
\end{array}
\]
where 
${\mbox{\tt t0}}$
that is poly-typable by 
$\mbox{\tt A0->B}$
is obtained from the Strong Typed B\"{o}hm Theorem (Theorem~\ref{thmStrongTypedBohmTheorem})
which represents a surjection from $\{ \mbox{\tt r1}, \mbox{\tt r2} \}$ to
one or two element set
$\{ \mbox{\tt u1} \, \, \mbox{\tt r1}, \, \mbox{\tt u1} \, \, \mbox{\tt r2}, \, \mbox{\tt u2} \, \, \mbox{\tt r1}, \, \mbox{\tt u2} \, \, \mbox{\tt r2}  \}$, 
the subterm $\mbox{\tt vi}$ is poly-typable by $\mbox{\tt B} \mbox{\tt ->} \mbox{\tt B}$
obtained from Constant Function Lemma
for each $i \, (1 \le i \le n)$,
and 
the subterm $\mbox{\tt wj}$ 
has type $\overbrace{\mbox{\tt C1->} \cdots \mbox{\tt -> Ckj}}^{k_j} \mbox{\tt ->D}$ 
for each $j \, (1 \le j \le m)$ where
$\mbox{\tt Ci}$ and $\mbox{\tt D}$ is poly-typable by $\mbox{\tt B} \mbox{\tt ->} \mbox{\tt B}$.
The subterm $\mbox{\tt wj}$ 
is constructed 
from Projection Lemma w.r.t. an appropriate position 
except for $G$ and $H$.
For example, in Figure~\ref{fig-Memo-Figure-4}, the first projection function is assigned to $F$. 
The terms $G$ and $H$ are constructed by the following two steps:
\begin{enumerate}
\item 
First we construct a term $\mbox{\tt mj}$ with type $\mbox{\tt D}$ using the Strong Typed B\"{o}hm Theorem (Theorem~\ref{thmStrongTypedBohmTheorem}). 
The functions for $G$ and $H$ are
the constant, identity, or negation functions, 
depending on $\mbox{\tt u1}$ and $\mbox{\tt u2}$.
Note that in order to represent the negation function we need the Strong Typed B\"{o}hm Theorem. 
\item 
Second from using $\mbox{\tt mj}$, we construct $\mbox{\tt wj}$ using Constant Function Lemma in order to discard
the unnecessary arguments.
\end{enumerate}
\begin{figure}[htbp]
 \begin{center}
  \includegraphics[scale=0.5]{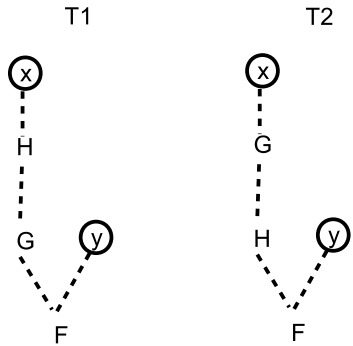}
 \end{center}
 \caption{Two different linearly labeled trees (4)}
 \label{fig-Memo-Figure-4}
\end{figure}
\end{itemize}
\item Otherwise: \\
The case is any of the degenerated versions of the cases above. 
We can apply the same discussion. 
\end{itemize}
$\Box$
\end{proof}

\section{Functional Completeness of Linear Types: An Application of Strong Typed B\"{o}hm Theorem}
Strong typed B\"{o}hm theorem for the linear $\lambda$-calculus is not a theoretical non-sense.
It has an algorithmic content and at least one application: functional completeness of linear types.
\begin{definition}
\label{defFunComp}
Let ${\tt A}$ be a type that has two different closed terms ${\tt s1}$ and ${\tt s2}$.
A function $f : { \{ 0, 1 \} }^{n} \to \{ 0, 1 \}$ is represented by 
a closed term ${\tt t}$ that is poly-typable by
$\overbrace{{\tt A} \mbox{\tt ->} \, \, \cdots \mbox{\tt ->} {\tt A}}^{n} \mbox{\tt ->} {\tt A}$
with regard to ${\tt s1}$ and ${\tt s2}$ 
if, for any $\langle x_1, \ldots, x_n \rangle \in { \{ 0, 1 \} }^{n}$ and $y \in \{ 0, 1 \}$ 
\[
f (x_1, \ldots, x_n) = y \, \, \Leftrightarrow \, \,
{\tt t} \, \, {\tt x1} \cdots {\tt xn} =_{\beta \eta {\rm c}} {\tt y}
\]
where
${\tt x1}, \ldots, {\tt xn}, {\tt y}$ are the images
of $x_1, \ldots, x_n, y$ under the map $\{ 0 \mapsto {\tt s1}, \, 1 \mapsto {\tt s2} \}$ respectively.
The type ${\tt A}$ is functionally complete with regard to ${\tt s1}$ and ${\tt s2}$ if
any function $f : { \{ 0, 1 \} }^{n} \to \{ 0, 1 \}$ is represented by 
a closed term with regard to ${\tt s1}$ and ${\tt s2}$.
\end{definition}
The following proposition is well-known. 
\begin{proposition}
\label{propDefFunComp}
A type ${\tt A}$ is functionally complete if and only if 
the Boolean {\it not} gate, the {\it and} gate, and 
the {\it duplicate} function, i.e., $\{ 0 \mapsto \langle 0, 0 \rangle, \, 1 \mapsto \langle 1, 1 \rangle \}$ 
are represented over ${\tt A}$. 
\end{proposition}
So far Mairson \cite{Mai04} gave
the functional completeness 
of
type 
$\MBB{B}_{\rm HM} 
= 
\mbox{\tt 'a} \mbox{\tt ->} \mbox{\tt 'a} \mbox{\tt ->} 
\mbox{\tt (} \mbox{\tt 'a} \mbox{\tt ->} \mbox{\tt 'a} \mbox{\tt ->} \mbox{\tt 'a} \mbox{\tt )}
\mbox{\tt ->} \mbox{\tt 'a}$
with regard to the two closed terms.
Moreover 
van Horn and Mairson \cite{vanHM07}
gave the functional completeness
of
$\MBB{B}_{\rm TWIST} \mbox{\tt *} \MBB{B}_{\rm TWIST}$
with regard to its two closed terms, 
where $\MBB{B}_{\rm TWIST} = \mbox{\tt 'a} \mbox{\tt *} \mbox{\tt 'a} \mbox{\tt ->} \mbox{\tt 'a} \mbox{\tt *} \mbox{\tt 'a}$.
In fact, the following theorem holds. 
\begin{theorem}
\label{thmFunctionalCompleteness}
Let ${\tt A}$ be a type that has two different closed terms ${\tt s1}$ and ${\tt s2}$.
Then the type ${\tt A}$ is functionally complete with regard to ${\tt s1}$ and ${\tt s2}$.
\end{theorem}
\begin{proof}
The representability of the {\it not} gate and the duplicate function are a direct consequence of strong typed B\"{o}hm theorem:
while in the {\it not} gate we choose $\mbox{\tt A}$ as $\mbox{\tt B}$ in Theorem~\ref{thmStrongTypedBohmTheorem}
and ${\tt s2}$ and ${\tt s1}$ as ${\tt u1}$ and ${\tt u2}$ respectively, 
in the duplicate function we choose $\mbox{\tt A} \mbox{\tt *} {\tt A}$ as $\mbox{\tt B}$
and $\mbox{\tt (} \mbox{\tt s1} \mbox{\tt ,} \mbox{\tt s1} \mbox{\tt )}$ and $\mbox{\tt (} \mbox{\tt s2} \mbox{\tt ,} \mbox{\tt s2} \mbox{\tt )}$
as ${\tt u1}$ and ${\tt u2}$ respectively.

On the other hand, by Constant Function Lemma (Lemma~\ref{lemmaConstantFunctionLemma}), 
there is a term ${\tt t}$ with poly-type $\mbox{\tt A} \mbox{\tt ->} \mbox{\tt A}$ that 
represents the constant function $\{ 0 \mapsto 0, \, 1 \mapsto 0 \}$.
Then we choose $\mbox{\tt A} \mbox{\tt ->} \mbox{\tt A}$ as $\mbox{\tt A0} \mbox{\tt ->} \mbox{\tt B}$ in Theorem~\ref{thmPolytimeVersionOfStrongTypedBohmTheorem}
and we choose $\mbox{\tt t}$ and $\mbox{\tt I} = \mbox{\tt fn} \, \, \mbox{\tt x} \mbox{\tt =>} \mbox{\tt x}$ as $\mbox{\tt u1}$ and $\mbox{\tt u2}$ respectively.
Then we get a term $\mbox{\tt t'}$ that represents the {\it and} gate. 
$\Box$
\end{proof}
Appendix~\ref{appFC-BHM} gives a functional completeness proof of $\MBB{B}_{\rm HM}$, which 
is extracted from proofs shown above
and is slightly different from that of \cite{Mai04}.
Note that our construction of functional completeness is not compatible with the polymorphic $\lambda$-calculus by Girard and Reynolds (for example, see \cite{GLT89,Cro94}):
For example, {\tt Not\_HM} can not be typed by 
$\forall \mbox{\tt 'a}. \MBB{B}_{\rm HM} \, \mbox{\tt ->} \, \forall \mbox{\tt 'a}. \MBB{B}_{\rm HM}$.
As far as we know, the only type that is compatible with the polymorphic $\lambda$-calculus is 
$\MBB{B}_{\rm Seq} = \mbox{\tt 'a} \mbox{\tt ->} \mbox{\tt(} \mbox{\tt 'a} \mbox{\tt ->} \mbox{\tt 'a} \mbox{\tt )} \mbox{\tt ->} \mbox{\tt (} \mbox{\tt 'a} \mbox{\tt ->} \mbox{\tt 'a} \mbox{\tt )} \mbox{\tt ->} \mbox{\tt 'a}$.
Appendix~\ref{appFC-BSeq} gives the functional completeness proof of $\MBB{B}_{\rm Seq}$ that is compatible with the polymorphic lambda calculus.
While the encoding derived from our proof of Theorem~\ref{thmFunctionalCompleteness} is not compatible with the calculus,
the modified version given in Appendix~\ref{appFC-BSeq} is compatible. 
It would be interesting to pursue this topic, i.e., whether or not other types are compatible with the polymorphic $\lambda$-calculus.

\section{Concluding Remarks}
With regard to the functional completeness problem of the linear $\lambda$-calculus, Theorem~\ref{thmFunctionalCompleteness} is not the end of the story.
For example, we have already found some better Boolean encodings than that given by Theorem~\ref{thmFunctionalCompleteness} (see Appendix~\ref{appFC-BSeq} and \cite{Mat15}).
We should discuss efficiency of various Boolean encodings in the linear $\lambda$-calculus and 
relationships among them. 
Moreover the extension to $n$-valued cases instead of the $2$-valued Boolean case is open. 
Our result is the first step toward these research directions.

\begin{ack}
  The author thanks an anonymous referee, who pointed out 
  the simplified definition of the relation $\leftrightarrow_{\rm c}$. 
\end{ack}
\nocite{*}
\bibliographystyle{eptcs}
\bibliography{final-eptcs-StrongTypedBohmTheorem}

\appendix
\section{The relationship between linear $\lambda$ terms and IMLL proof nets}
\subsection{Brief Introduction to IMLL proof nets}
\label{appIMLLProofNets}
In this appendix, we introduce IMLL proof nets briefly.
For a complete treatment, for instance see \cite{Mat07}.
\begin{definition}[Plain and signed IMLL formulas]
The plain IMLL formulas are defined in the following grammar:
\[
A ::= p \, | \, A \TENS B \, | \, A \LIMP B
\]
where $p$ is called a propositional variable.
A signed IMLL formula has the form $A^+$ or $A^-$, where $A$ is a plain IMLL formula. 
\end{definition}
\begin{definition}[Links]
A link is an object with a few signed IMLL formulas. 
Any link is any of ID-, $\TENS^+$-, $\TENS^-$-, $\LIMP^+$-, or $\LIMP^-$-link shown in Figure~\ref{figLinks}.
\end{definition}
\begin{figure}[htbp]
\begin{center}
  \includegraphics[scale=0.6]{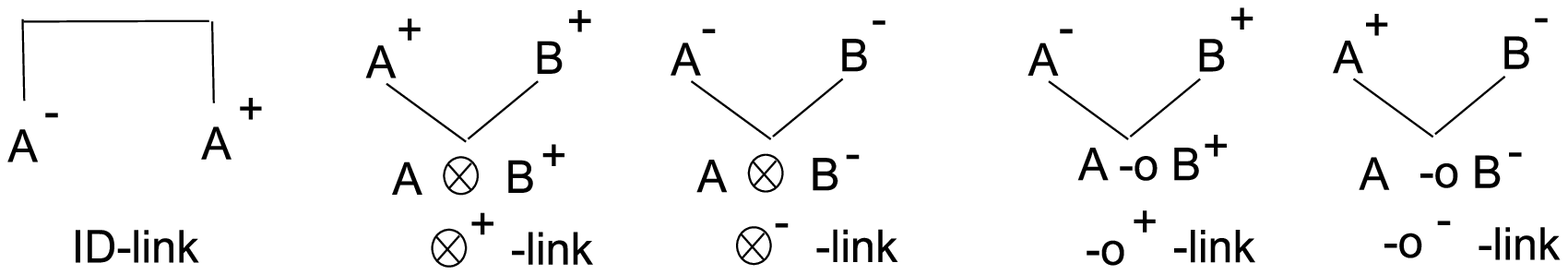}
\end{center}
 \caption{Links}
 \label{figLinks}
\end{figure}
\begin{definition}[IMLL proof nets]
An IMLL proof net is defined inductively as shown in Figure~\ref{figIMLLProofNets}.
\end{definition}
\begin{figure}[htbp]
\begin{center}
  \includegraphics[scale=0.6]{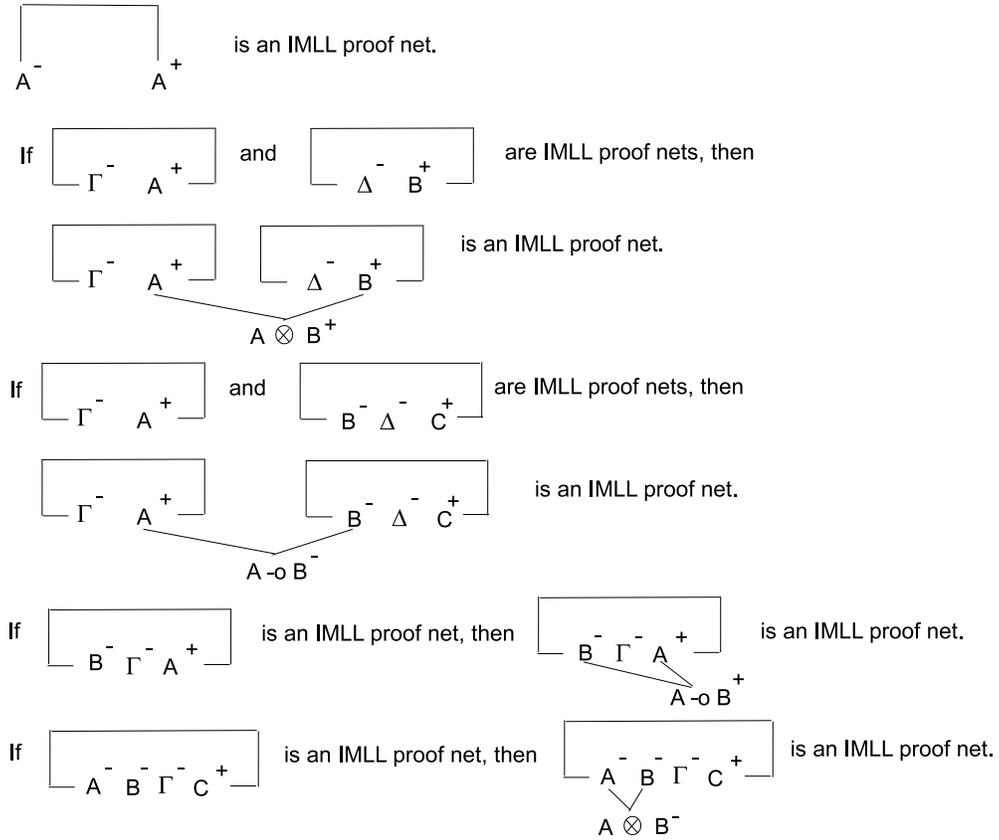}
\end{center}
 \caption{IMLL proof nets}
 \label{figIMLLProofNets}
\end{figure}
\begin{definition}[Reduction rules]
Reduction rules for an IMLL proof net have two kinds:
one is multiplicative shown in Figure~\ref{figIMLLReduction}
and
the other $\eta$ shown in Figure~\ref{figIEtaReduction}.
\end{definition}
\begin{figure}[htbp]
\begin{center}
  \includegraphics[scale=0.6]{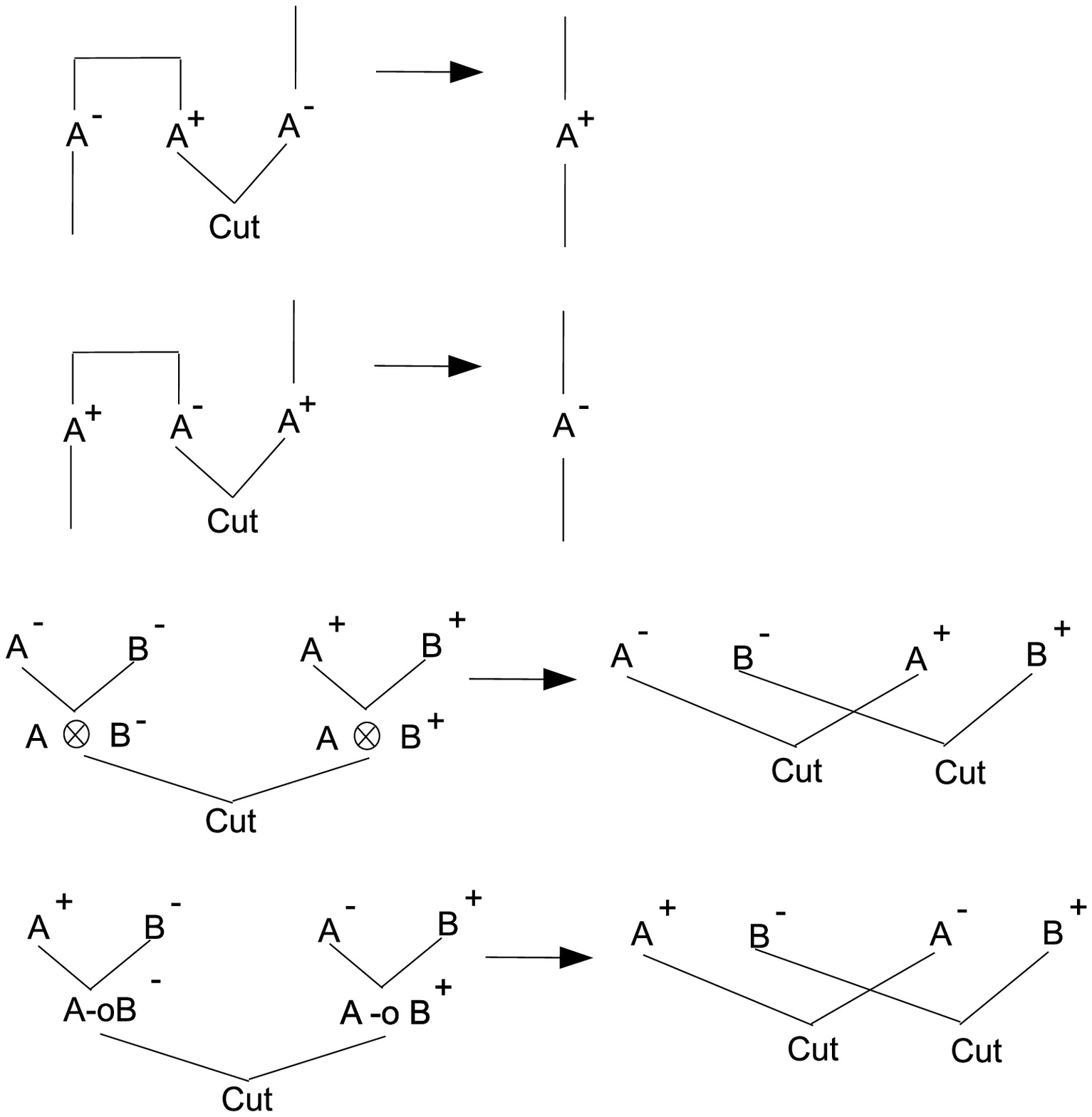}
\end{center}
 \caption{Multiplicative reduction rules}
 \label{figIMLLReduction}
\end{figure}
\begin{figure}[htbp]
\begin{center}
  \includegraphics[scale=0.6]{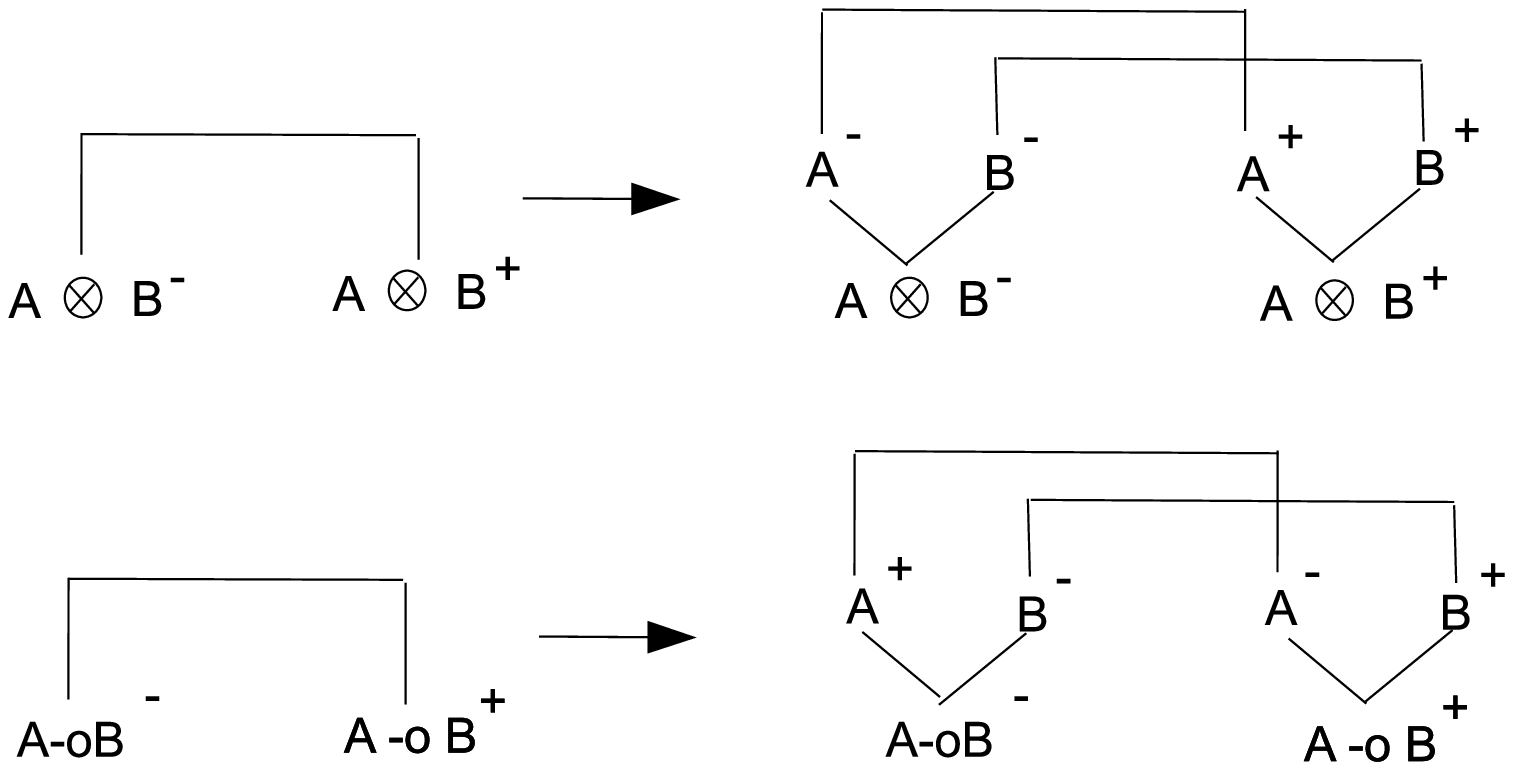}
\end{center}
 \caption{$\eta$ reduction rules}
 \label{figIEtaReduction}
\end{figure}
The reduction relation over IMLL proof nets induced by these reduction rules is strong normalizing and confluent.
So we can obtain a unique normal form of any IMLL proof net.
For two IMLL proof nets $\Theta_1$ and $\Theta_2$,
$\Theta_1$ is equal to $\Theta_2$ (denoted by $\Theta_1 = \Theta_2$)
if there is a bijective map from the signed IMLL formula occurrences in the normal form of $\Theta_1$ to that of $\Theta_2$
such that the map preserves the link structure (for the complete treatment, see \cite{Mat07}). 
\subsection{A full and faithful embedding of the linear $\lambda$-calculus into IMLL proof nets}
First we define our translation $\llbracket - \rrbracket$ of linear $\lambda$-terms into IMLL proof nets by
Figure~\ref{figTransLLC}, where we identify IMLL proof nets up to $=$ defined by Definition 14 in \cite{Mat07}
(or Appendix~\ref{appIMLLProofNets}). 
\begin{figure}[htbp]
\begin{center}
  \includegraphics[scale=0.6]{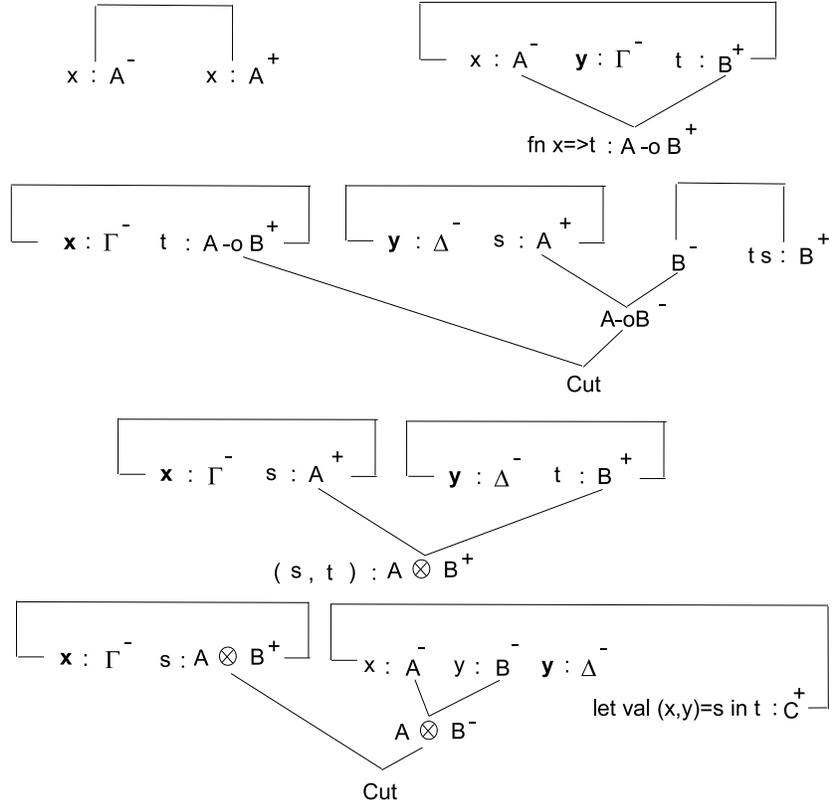}
\end{center}
 \caption{Translation of Linear $\lambda$-Terms into IMLL proof nets}
 \label{figTransLLC}
\end{figure}
Then the following proposition holds.
\begin{proposition}
\label{propTransPreservBetaEta-C}
If $\mbox{\tt t} \to_{\beta \eta {\rm c}} \mbox{\tt t'}$  then, $\llbracket \mbox{\tt t} \rrbracket = \llbracket \mbox{\tt t'} \rrbracket$.
\end{proposition}
\begin{proof}
When $\mbox{\tt t} \to_{\beta} \mbox{\tt t'}$, Figure~\ref{figBetaRedexes} proves the proposition.
When $\mbox{\tt t} \to_{\eta} \mbox{\tt t'}$, we consider Figure~\ref{figEtaRedexes}.
In Figure~\ref{figEtaRedexes},  we normalize IMLL proof nets $\Theta$ and $\Pi$. 
Then the proposition should be obvious. 
When $\mbox{\tt t} \leftrightarrow_{\rm c} \mbox{\tt  t'}$, $\mbox{\tt t}$ and $\mbox{\tt t'}$ are translated into the same IMLL proof net in each case.
$\Box$
\end{proof}
\begin{figure}[htbp]
\begin{center}
  \includegraphics[scale=0.6]{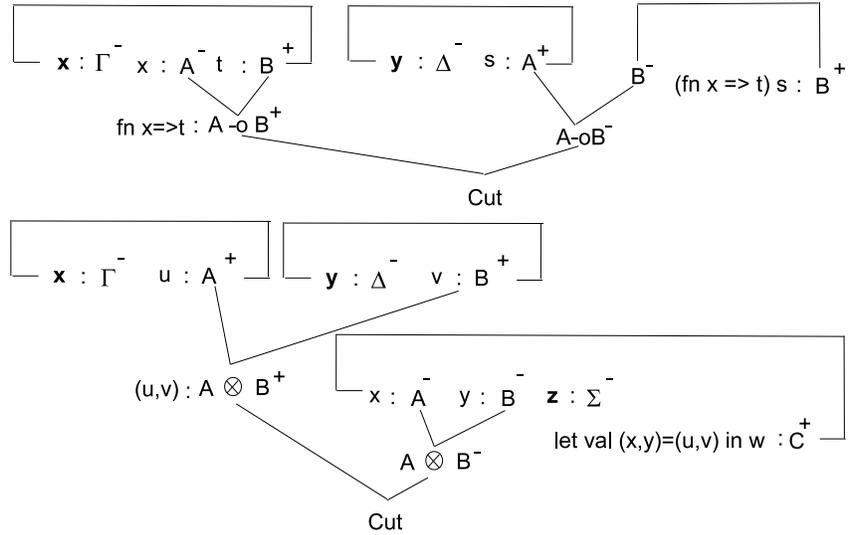}
\end{center}
 \caption{Translation of $\beta$-redexes}
 \label{figBetaRedexes}
\end{figure}
\begin{figure}[htbp]
\begin{center}
  \includegraphics[scale=0.6]{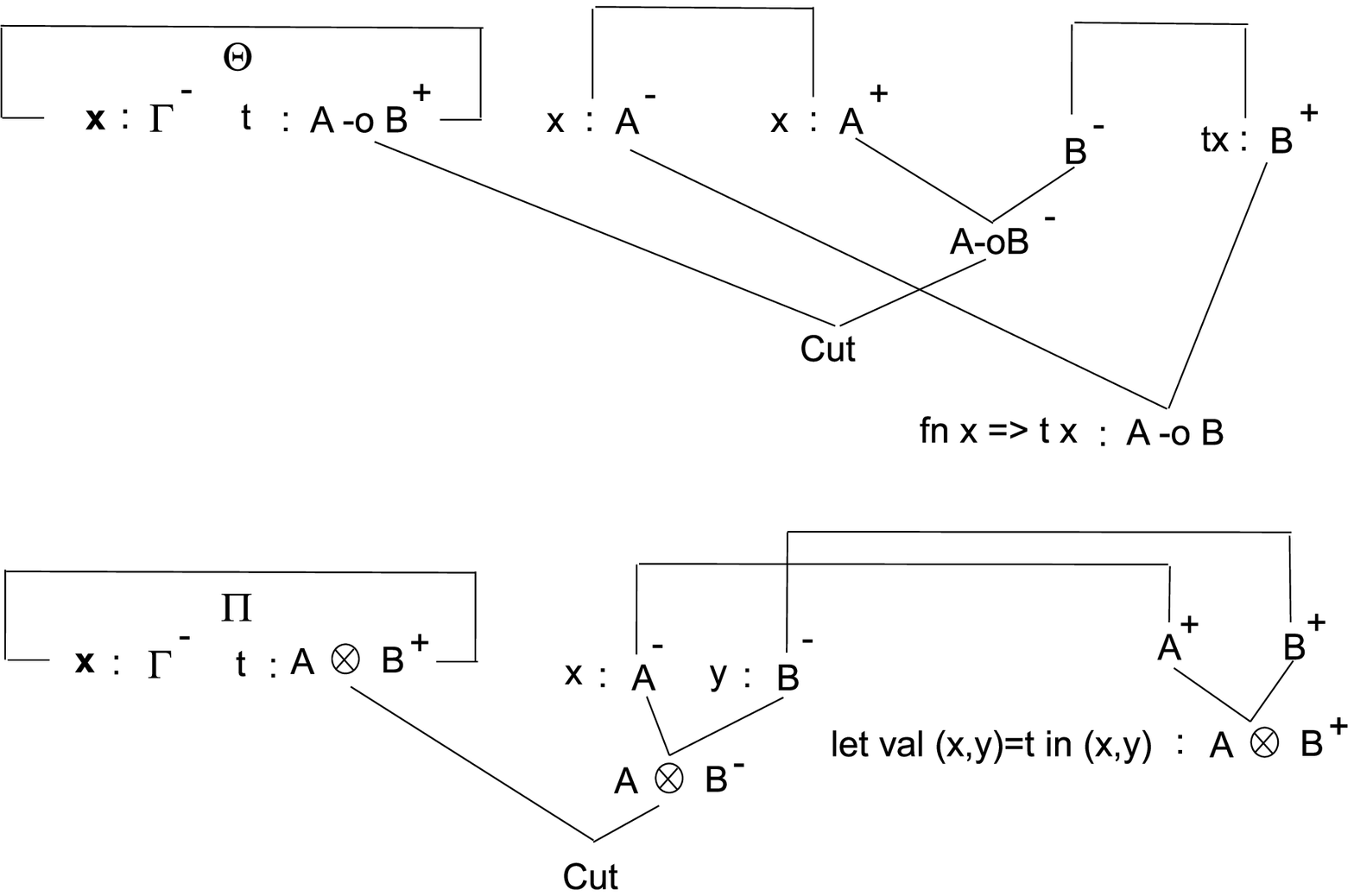}
\end{center}
 \caption{Translation of $\eta$-redexes}
 \label{figEtaRedexes}
\end{figure}
Moreover if both $\mbox{\tt t}$ and $\mbox{\tt t'}$ are normal forms of linear $\lambda$-terms with regard to $\to_{\beta \eta {\rm c}}$, then
when $\neg (\mbox{\tt t} =_{\beta \eta {\rm c}} \mbox{\tt t'})$, it is obvious that $\llbracket \mbox{\tt t} \rrbracket \neq \llbracket \mbox{\tt t'} \rrbracket$.
So we have established the faithfulness. 
On the other hand, for any IMLL proof net $\Theta$ whose conclusion is a type of the linear $\lambda$-calculus, 
it is easy to show that there is a linear $\lambda$-term $\mbox{\tt t}$ such that $\llbracket \mbox{\tt t} \rrbracket = \Theta$.
So we have established the fullness. 
Therefore we conclude the existence of a full and faithful embedding stated above.
So we can identify a normal linear $\lambda$-term with the corresponding normal IMLL proof net. 
We treat $=_{\beta \eta {\rm c}}$ as the legitimate equality of linear $\lambda$-terms. 
Note that while $\eta$-normal forms are natural in the linear $\lambda$-calculus, 
$\eta$-long normal forms are natural in the proof net formalism. 

\section{Why Need Poly-Types?}
\label{secWhyNeedPoly-Types}
In this appendix, we show that the method of \cite{Mat07} can not be extended without poly-types.

We let
$\MBB{B}_{\rm HM} 
= 
\mbox{\tt 'a} \mbox{\tt ->} \mbox{\tt 'a} \mbox{\tt ->} 
\mbox{\tt (} \mbox{\tt 'a} \mbox{\tt ->} \mbox{\tt 'a} \mbox{\tt ->} \mbox{\tt 'a} \mbox{\tt )}
\mbox{\tt ->} \mbox{\tt 'a}$
and 
$\MBB{B}_{\rm Seq} = \mbox{\tt 'a} \mbox{\tt ->} \mbox{\tt(} \mbox{\tt 'a} \mbox{\tt ->} \mbox{\tt 'a} \mbox{\tt )} \mbox{\tt ->} \mbox{\tt (} \mbox{\tt 'a} \mbox{\tt ->} \mbox{\tt 'a} \mbox{\tt )} \mbox{\tt ->} \mbox{\tt 'a}$
and \\
{\tt fun True x y z = z x y;} \\
{\tt fun False x y z = z y x;} \\
{\tt fun TrSeq x f g = g (f x);} \\
{\tt fun FlSeq x f g = f (g x);} \\
The terms {\tt True} and {\tt False} are closed terms of $\MBB{B}_{\rm HM}$
and {\tt TrSeq} and {\tt FlSeq} are that of $\MBB{B}_{\rm Seq}$.
Then we show that
for any type {\tt A}, we cannot find a closed term {\tt s} of type
$\MBB{B}_{\rm Seq}[\mbox{\tt A}/\mbox{\tt 'a}] \mbox{\tt ->} \MBB{B}_{\rm HM}$
such that
\[
\mbox{\tt s} \, \, \mbox{\tt TrSeq} =_{\beta \eta {\rm c}} \mbox{\tt True}
\quad \mbox{and}
\quad
\mbox{\tt s} \, \, \mbox{\tt FlSeq} =_{\beta \eta {\rm c}} \mbox{\tt False} \, \, .
\]
We suppose that there is such a closed term {\tt s}.
Then ${\tt A}$ must be $\MBB{B}_{\rm HM}$.
Moreover there must be  closed terms {\tt f} and {\tt g} of type
$\MBB{B}_{\rm HM} \mbox{\tt ->} \MBB{B}_{\rm HM}$ such that
\[
\mbox{\tt f} \mbox{\tt (} \mbox{\tt g} \, \mbox{\tt t} \mbox{\tt )}
=_{\beta \eta {\rm c}} {\tt True}
\quad \mbox{and}
\quad
\mbox{\tt g} \mbox{\tt (} \mbox{\tt f} \, \mbox{\tt t} \mbox{\tt )}
=_{\beta \eta {\rm c}} {\tt False}
\]
where {\tt t} is {\tt True} or {\tt False}.
But {\tt f} and {\tt g} must be {\it identity} or {\it not gate},
because $\MBB{B}_{\rm HM} \mbox{\tt ->} \MBB{B}_{\rm HM}$ does not allow any constant functions.
This is impossible. 

\section{Functional Completeness of $\MBB{B}_{\rm HM}$}
\label{appFC-BHM}
The terms {\tt Not\_HM}, {\tt Copy\_HM}, \, {\tt And\_HM} below are derived from our construction.\\
{\tt fun True x y z = z x y;} \\
{\tt fun False x y z = z y x;} \\
{\tt fun I x = x;} \\
{\tt fun u\_2 x1 x2 = x1 (x2 I);} \\
{\tt fun u\_3 x1 x2 x3 = x1 (x2 (x3 I));}\\
{\tt fun proj\_1 x1 x2 = x2 I I u\_2 x1;}\\
{\tt fun Not\_HM x = x False True proj\_1;} \\
{\tt fun LDTr\_Pair p x y f z w h l}\\
\ \ \ \ \ \ \ \ \ \ \ {\tt = let val (u,v) = p in l (u x y f) (v z w h) end;} \\
{\tt fun proj\_Pair\_1 x1 x2  = LDTr\_Pair x2 I I u\_2 I I u\_2 u\_2 x1;} \\
{\tt fun Copy\_HM x = x (True,True) (False,False) proj\_Pair\_1;}\\
{\tt fun const\_F x = x I I (u\_2) False;}\\
{\tt fun And\_HM x y = let val (u,v) = Copy\_HM y in} \\
\ \ \ \ \ \ \ \ \ \ \ \ \ \ \ \ \ \ \ \ \ \ \ \ \ \ \ \ \ \ \ \ \ {\tt x (I u) (const\_F v) proj\_1 end;}

\section{Functional Completeness of $\MBB{B}_{\rm seq}$}
\label{appFC-BSeq}
The terms {\tt NotSeq}, {\tt CopySeq}, {\tt AndSeq} below are compatible with the polymorphic lambda calculus of Girard-Reynolds. \\
{\tt fun TrSeq x f g = g (f x);} \\
{\tt fun FlSeq x f g = f (g x);} \\
{\tt fun NotSeq h x f g = h x g f;} \\
{\tt fun constTr h x f g = g (f (h x I I));} \\
{\tt fun conv h z = let val (f,g) = h in let val (x,y) = z} \\
\ \ \ \ \ {\tt in (f x,g y) end end;}\\
{\tt fun CopySeq x = } \\
\ \ \ \ {\tt x (TrSeq,TrSeq) (conv (NotSeq,NotSeq)) (conv (constTr,constTr));} \\
{\tt fun constFlFun h k x f g = f (g (k (h FlSeq x I I) I I));}\\
{\tt fun idFun h k x f g = k (h TrSeq x I I) f g;}\\
{\tt fun AndSeq x = x I constFlFun idFun;} 

\end{document}